\documentclass[journal]{IEEEtran}
\usepackage{amsfonts}
\usepackage{mathrsfs}
\usepackage{multirow}
\usepackage{mdwmath}
\usepackage{balance}
\usepackage{color}
\usepackage{citesort}
\usepackage{url}
\usepackage{threeparttable}

\usepackage{bbm}
\usepackage{amsmath}
\usepackage{multirow}
\usepackage{booktabs}
\usepackage{setspace}

\usepackage{amssymb,amscd,amsmath}
\usepackage[all]{xy}
\usepackage{fancyhdr}
\usepackage{subfigure}
\usepackage{citesort}
\usepackage{algorithm,algorithmic}
\usepackage{stfloats}
\usepackage{graphicx}
\usepackage{slashbox}
\usepackage{stfloats}
\ifCLASSINFOpdf
\else
\fi

\newtheorem{theorem}{Theorem}
\newtheorem{definition}{Definition}
\newtheorem{lemma}{Lemma}
\newtheorem{corollary}{Corollary}
\DeclareMathOperator*{\argmax}{arg\,max}
\DeclareMathOperator*{\argmin}{arg\,min}
\newcommand{\vecnorm}[1]{\lVert #1 \rVert}
\newcommand{\e}{\mathrm{e}}

\begin{document}
\title{Confined Orthogonal Matching Pursuit for Sparse Random Combinatorial Matrices}

%\author{{Authors}
%
%}

%
%\author{{Author~1 and Author~2}
%}

\author{{Xinwei Zhao, Jinming Wen, Hongqi Yang, and Xiao Ma}
\thanks{This work was supported by the National Key R\&D Program of China (No.  2021YFA1000500), the National Natural Science Foundation of China (Nos. 62471506, 12271215, 12326378 and 11871248), and the Guangdong Provincial Key Laboratory of Information Security Technology (No.  2023B1212060026).~(\emph{Corresponding author:~Xiao Ma})

X.~Zhao, H.~Yang, and X.~Ma are with the School of Computer Science and Engineering, Sun Yat-sen University, Guangzhou, 510006, China, and X.~Ma is also with Guangdong Key Laboratory of Information Security Technology, Sun Yat-sen University, Guangzhou, 510006, China.~(zhaoxw9@mail2.sysu.edu.cn;~mcsyhq@mail.sysu.edu.cn;~maxiao@mail.sysu.edu.cn)

J.~Wen is with the School of Mathematics, Jilin University, Jilin, 130012, China.~(jinming.wen@mail.mcgill.ca)

%X.~Ma is with Guangdong Key Laboratory of Information Security Technology, Sun Yat-sen University, Guangzhou, 510006, China.~(maxiao@mail.sysu.edu.cn)

}
}

\maketitle

\begin{abstract}
Orthogonal matching pursuit~(OMP) is a commonly used greedy algorithm for recovering sparse signals from compressed measurements. In this paper, we introduce a variant of the OMP algorithm to reduce the complexity of reconstructing a class of $K$-sparse signals $\boldsymbol{x} \in \mathbb{R}^{n}$ from measurements $\boldsymbol{y} = \boldsymbol{A}\boldsymbol{x}$. In particular, $\boldsymbol{A} \in \{0,1\}^{m \times n}$ is a sparse random combinatorial matrix with independent columns, where each column is chosen uniformly among the vectors with exactly $d~(d \leq m/2)$ ones. The proposed algorithm, referred to as the confined OMP algorithm, leverages the properties of the sparse signal $\boldsymbol{x}$ and the measurement matrix $\boldsymbol{A}$ to reduce redundancy in $\boldsymbol{A}$, thereby requiring fewer column indices to be identified. To this end, we first define a confined set $\Gamma$~(see Definition~\ref{def:conf_set} in Sec.~\ref{subsec:confined_pro}) with $|\Gamma| \leq n$ and then prove that the support of $\boldsymbol{x}$ is a subset of $\Gamma$ with probability 1 if the distributions of nonzero components of $\boldsymbol{x}$ satisfy a certain condition. During the process of the confined OMP algorithm, the possibly chosen column indices are strictly confined to the confined set $\Gamma$. We further develop the lower bound on the probability of exact recovery of $\boldsymbol{x}$ using the confined OMP algorithm. The derived bound shows that, if the spark of $\boldsymbol{A}$ is greater than $K$, then $m = 2\e K \ln (n - K)$ measurements are sufficient to ensure the probability of recovering a $K$-sparse signal with the proposed algorithm is at least $1 - \frac{1}{n-K}$, where $\e$ represents Euler’s number. Furthermore, the obtained theoretical results can be used to optimize the column degree $d$ of $\boldsymbol{A}$. Finally, experimental results show that the confined OMP algorithm is more efficient in reconstructing a class of sparse signals compared to the OMP algorithm.

\end{abstract}


\section{Introduction}\label{introduction}
Compressed sensing~(CS) as a novel sampling theory~\cite{Candes2005Decoding,Donoho2006Compressed} has attracted much attention over the past twenty years. In CS, it is common to encounter the following linear model
\begin{equation}\label{equ:linear_mod}
\boldsymbol{y}  = \boldsymbol{A} \boldsymbol{x},
\end{equation}
where $\boldsymbol{A} \in \mathbb{R}^{m\times n}$ is a measurement matrix with $m < n$, $\boldsymbol{x} \in \mathbb{R}^{n}$ is an unknown $K$-sparse signal~(i.e., $\boldsymbol{x}$  has at most $K$ nonzero elements) and $\boldsymbol{y} \in \mathbb{R}^{m}$ is a known measurement vector. The CS recovery algorithm aims to recover $\boldsymbol{x}$ from~(\ref{equ:linear_mod}) by solving the following $l_0$-minimization problem
\begin{equation}\label{equ:L0_min}
  \min \vecnorm{\boldsymbol{x}}_0~~\mathrm{s.t.}~\boldsymbol{y}  = \boldsymbol{A} \boldsymbol{x},
\end{equation}
where the $l_0$ norm $\vecnorm{\boldsymbol{x}}_0$ is the number of nonzero elements in $\boldsymbol{x}$, i.e., $\vecnorm{\boldsymbol{x}}_0 \overset{\triangle}{=} |\{ i : x_i \neq 0 \}|$. Here, $x_i$ is the $i$-th element of $\boldsymbol{x}$ and $|\mathcal{S}|$ is the cardinality of a set $\mathcal{S}$. Unfortunately, as shown in~\cite{Foucart2013A}, the problem~(\ref{equ:L0_min}) is NP-hard in general. Two methods are commonly used for solving this problem. One focuses on the reconstruction of sparse signals by considering a convex relaxation of~(\ref{equ:L0_min}), such as solving the $l_1$-minimization problem~\cite{Candes2005Robust,Wainwright2009Sharp,Chen2001Atomic,Bradley2004Least}. There are also algorithms solving~(\ref{equ:L0_min}) directly, such as greedy algorithms~\cite{Tropp2007Signal,Needell2010ROMP,Dai2009SP,NEEDELL2009CoSaMP,Rubinstein2008Efficient,Liu2012OMMP,Wang2012GOMP,Kwon2014MMP,Wen2021BMP} and thresholding algorithms~\cite{Thomas2009Iterative,Foucart2011Hard,Tanner2013Normalized,Donoho1995Denoising,Beck2009A}.

\begin{algorithm}[t]
\caption{Orthogonal matching pursuit algorithm}
\label{alg:OMP}
\begin{algorithmic}
\STATE{Input: $\boldsymbol{y} \in \mathbb{R}^{m}$, $\boldsymbol{A} \in \mathbb{R}^{m \times n}$, $\varepsilon$, and $K$. }
\STATE{Initialize: $k =0$, $\boldsymbol{r}^{(0)} = \boldsymbol{y}$, and $\Lambda^{(0)} = \emptyset$. }
\WHILE{`$\vecnorm{\boldsymbol{r}^{(k)}}_2  > \varepsilon $ \textbf{and} $k < K$ are met' }
\STATE{$k=k+1$,}
\STATE{$t^{(k)}= \underset{i \in [n] \backslash \Lambda^{(k-1)}}\argmax |\boldsymbol{A}_i^{T}\boldsymbol{r}^{(k-1)} |$,~\textbf{(Identification)}}
\STATE{$\Lambda^{(k)} = \Lambda^{(k-1)} \cup \{t^{(k)} \}$,~\textbf{(Augmentation)}}
\STATE{$\hat{\boldsymbol{x}}_{\Lambda^{(k)}}^{(k)} = \underset{\boldsymbol{x} \in \mathbb{R}^{k}}\argmin \vecnorm{\boldsymbol{y} - \boldsymbol{A}_{\Lambda^{(k)}} \boldsymbol{x} }_2$,~\textbf{(Estimation)}}
\STATE{$\boldsymbol{r}^{(k)} = \boldsymbol{y} - \boldsymbol{A}_{\Lambda^{(k)}} \hat{\boldsymbol{x}}_{\Lambda^{(k)}}^{(k)}$.~\textbf{(Residual update)}}
\ENDWHILE
%\RETURN $T$
\STATE{Output: $\hat{\boldsymbol{x}}^{(k)}$. }
\end{algorithmic}
\end{algorithm}

Among greedy algorithms, the OMP algorithm~\cite{Tropp2007Signal} is one of the most commonly used algorithms. As a greedy algorithm, the OMP algorithm identifies the support (index set of nonzero elements) of the sparse signal $\boldsymbol{x}$ in an iterative manner and thus
iteratively performs local optimal updates. Specifically, the process of the OMP algorithm at each iteration can be summarized into four steps~\cite{Wang2012GOMP}~(see~Algorithm~\ref{alg:OMP} for details):
\begin{itemize}
  \item Identification: select the column of $\boldsymbol{A}$ maximally correlated with the residual $\boldsymbol{r}^{(k-1)}$.
  \item Augmentation: add the index of the chosen column into the estimated support set $\Lambda^{(k)}$.
  \item Estimation: estimate the values of elements whose indices are in the estimated support set $\Lambda^{(k)}$ by solving a least squares problem.
  \item Residual update: eliminate the vestige of columns in $\Lambda^{(k)}$ from the measurement vector $\boldsymbol{y}$, resulting in a new residual used for the next iteration.
  \item Stopping criteria: the iteration process is terminated when the number of iterations reaches the maximum value $K$ or the $l_2$ norm of the residual $\boldsymbol{r}^{(k)}$ falls below a preset threshold $\varepsilon$.
\end{itemize}
Among these, the computational complexity of the OMP algorithm is mainly dominated by the identification step and the estimation step. In order to enhance the computational efficiency and recovery performance of the OMP algorithm, there have been some studies on the modified OMP algorithm, mainly focusing on the identification step. For example, the generalized OMP algorithm~\cite{Wang2012GOMP}~(a.k.a. orthogonal multi-matching pursuit algorithm~\cite{Liu2012OMMP}) allows multiple indices maximally correlated with the residual to be chosen at each iteration so that fewer number of iterations are required. The existing methods seek to efficiently select the ``true" column indices from the redundant dictionary~(measurement matrix~$\boldsymbol{A}$). This raises a question: Is it possible to reduce the redundancy of the dictionary $\boldsymbol{A}$?

Before answering this question, we introduce some useful tools to characterize the performance of the recovery algorithm. In~\cite{Candes2005Decoding}, Cand\`{e}s and Tao introduced the concept of \emph{restricted isometry property}~(RIP) and showed that if $\boldsymbol{A}$ satisfies RIP with relatively small \emph{restricted isometry constant}~(RIC) $\delta_{2K}$, any $K$-sparse signal can be exactly recovered by solving a $l_1$-minimization problem. In particular, it has been proved in~\cite{Wen2017A} that $\delta_{K+1} < 1/\sqrt{K+1}$ is sufficient for the OMP algorithm to recover any $K$-sparse signal from~(\ref{equ:linear_mod}) in $K$ iterations. The \emph{mutual coherence}, denoted as $\mu_m$, is also an important parameter for $\boldsymbol{A}$, which indicates the maximum absolute correlation between normalized columns of $\boldsymbol{A}$. It has been shown in~\cite{Tropp2004Greed} that any $K$-sparse signal can be exactly recovered by the OMP algorithm if $K < (\mu_m^{-1} + 1 )/2$. In~\cite{Tropp2007Signal}, the authors developed a lower bound on the probability that any $K$-sparse signal can be exactly recovered from~(\ref{equ:linear_mod}) by using the OMP algorithm in $K$ iterations, where the matrix $\boldsymbol{A}$ in~(\ref{equ:linear_mod}) is a random Gaussian matrix. In~\cite{Wen2020Signal}, this lower bound is further tightened with the aid of prior information of $\boldsymbol{x}$. Unfortunately, the lower bound techniques proposed in~\cite{Tropp2007Signal} and~\cite{Wen2020Signal} are only suitable for the case of Gaussian matrices.

In CS, the construction of measurement matrices is also one of the main concerns. In general, the random measurement matrices can be classified into dense and sparse matrices. It has been verified that many dense matrices, such as Gaussian matrices and Fourier matrices, satisfy the RIP with overwhelming probability~\cite{Baraniuk2008A} and have provably good recovery performance. On the other hand, sparse random matrices also attract much attention~\cite{berinde2008sparse,Jafarpour2009Efficient,Khajehnejad2011Sparse,Dimakis2012LDPC,Liu2017Reconstruction,Lu2018Binary,Zhao2025Sparse} since the sparsity can enable the computation of the matrix-product to be remarkably efficient and save the storage space in practice~\cite{Gilbert2010Sparse}. In particular, sparse binary-valued measurement matrices are commonly used in some applications, including group testing~\cite{Asilomar2008Group}, DNA Microarrays~\cite{Parvaresh2008Recovering}, and single-pixel imaging~\cite{Duarte2008Single}. Furthermore, many studies~\cite{Dimakis2012LDPC,Liu2017Reconstruction,Lotfi2020Compressed} showed that the sparse binary-valued measurement matrices are as ``good" as the dense ones both in theory and in practice.

In this paper, we focus on a type of sparse binary measurement matrix known as a \emph{sparse random combinatorial matrix}. This matrix $\boldsymbol{A} \in \{0,1\}^{m\times n}$ is constructed such that each column is generated independently and uniformly at random from the set of all binary vectors of length $m$ that contain exactly $d$ ones, where $d \leq m/2$. In particular, sparse random combinatorial matrices have some favorable properties. It has been shown in~\cite{Ferber2022Singularity,aigner2022sparse} that an $m \times m$ sparse random combinatorial matrix is nonsingular with probability $1-o(1)$ for a sufficiently large $m$, if the degree $d$ of each column satisfies $(1+\gamma)\ln m \leq d \leq m/2$ for a constant $\gamma > 0$. Now we answer the previous question. Yes, much redundancy of $\boldsymbol{A}$ can be removed if $\boldsymbol{x}$ is a signal defined in~Definition~\ref{def:conf_signal}. Actually, this observation can be traced back to the work of Khajehnejad et al. in~\cite{Khajehnejad2011Sparse}. They showed that the redundancy of a sparse matrix constructed by the expander theory can be eliminated if nonzero elements of $\boldsymbol{x}$ are non-negative. Our work is more general and takes the observation in~\cite{Khajehnejad2011Sparse} as a special case. Specifically, the sparse signal we consider is not limited to be non-negative, but a class of signals defined in~Definition~\ref{def:conf_signal}, including the Gaussian sparse signal. Furthermore, the considered measurement matrix $\boldsymbol{A}$ in this paper is more general than that constructed by the expander theory in~\cite{Khajehnejad2011Sparse}. The contributions of this paper are summarized as follows.
\begin{itemize}
  \item We first define the confined set $\Gamma$ with $|\Gamma| \leq n$ and prove that the support of $\boldsymbol{x}$ defined in~Definition~\ref{def:conf_signal} is a subset of $\Gamma$ with probability 1. To theoretically clarify the effectiveness of removing the redundancy of $\boldsymbol{A}$, we present the expectations of the sparsity of $\boldsymbol{y}$ and the size of $\Gamma$.
  \item We propose a variant of the OMP algorithm, referred to as the confined OMP algorithm, by introducing the confined set $\Gamma$ into the identification step. The possibly chosen column indices are strictly confined to the confined set $\Gamma$. We further analyze the complexities of the OMP and confined OMP algorithms. The analysis results show that the identification efficiency of confined OMP algorithm is at least $\frac{nKd-K}{|\Gamma|Kd-K+n}$ times that of OMP algorithm. Furthermore, the experimental results show that the confined OMP algorithm achieves a large reduction in complexity if $K \ll m$.
  \item We develop a lower bound on the probability of exact recovery of $\boldsymbol{x}$ defined in~Definition~\ref{def:conf_signal} using the confined OMP algorithm over a sparse random combinatorial matrix. Theoretical results show that, if the spark of $\boldsymbol{A}$ is greater than $K$, $m = 2\e K \ln (n - K)$ measurements  are sufficient to guarantee that the probability of recovering a $K$-sparse signal using the proposed algorithm is no lower than $1 - \frac{1}{n-K}$. In the asymptotic regime where both $m$ and $n$ tend to infinity with $n = m^\tau~(\tau > 1)$, the proposed algorithm can exactly recover signals with sparsity $K = o\left(\frac{m}{\ln m}\right)$ with probability 1.
      %The experimental results show that the bounds of confined OMP are tighter than those of OMP.
\end{itemize}

% The outline is not required, but we show an example here.
The paper is organized as follows. We define the confined set $\Gamma$ and present the proposed algorithm in~Sec.~\ref{sec:confined}. In~Sec.~\ref{sec:analysis}, the expectations of the sparsity of $\boldsymbol{y}$ and the size of $\Gamma$ are investigated, and then the recovery performance analysis of the confined OMP algorithm is provided. Experimental results are presented in~Sec.~\ref{sec:experiments} and~Sec.~\ref{sec:conclusions} concludes the paper.

\emph{Notation:} we use boldface lowercase letters to denote column vectors and boldface
uppercase letters to denote matrices. The support of $\boldsymbol{x}$ is denoted by $\Omega$ and the complement of $\Omega$ is $\Omega^c = [n]\backslash \Omega = \{i:i\in [n],i \notin \Omega\}$, where $[n]$ represents the set $\{1,2,\cdots,n\}$. Let $x_i$ and $\boldsymbol{A}_j$ be the $i$-th element of $\boldsymbol{x}$ and the $j$-th column of $\boldsymbol{A}$, respectively. We use $|\Gamma|$ to denote the cardinality of a set $\Gamma$ and $A_{i,j}$ to denote the element of $\boldsymbol{A}$ located at the $i$-th row and the $j$-th column. We denote by $\boldsymbol{x}_{\Lambda}$ the sub-vector of $\boldsymbol{x}$ that contains the entries of $\boldsymbol{x}$ indexed by the set $\Lambda$, and $\boldsymbol{A}_{\Lambda}$ the sub-matrix of $\boldsymbol{A}$ that contains the columns of $\boldsymbol{A}$ indexed by the set $\Lambda$. We use $\boldsymbol{A}^T$ to denote the transpose of $\boldsymbol{A}$ and $\mathbf{E}[X]$ to represent the expectation of $X$.

\section{Proposed Algorithm}\label{sec:confined}
This section will first introduce the inherent properties of random combinatorial matrices and sparse signals, and then present the details of the proposed algorithm.

\subsection{The Confined Set}\label{subsec:confined_pro}
For a real number $\epsilon > 0$, let $\mathcal{E} = \{i: |y_i| \leq \epsilon \}$ where $y_i$ is the $i$-th element of $\boldsymbol{y}$. The definition of the confined set is as follows.
\begin{definition}[The confined set]\label{def:conf_set}
The confined set $\Gamma$ is specified by defining its complement $\Gamma^c$ as $\Gamma^c = \bigcup_{i\in \mathcal{E}}\Gamma^c_{i}$ where $\Gamma_{i}^{c} = \{ j : A_{i,j} = 1,~|y_i| \leq \epsilon \}$. That is, $\Gamma = [n]\backslash \Gamma^c$.
%Let $\Gamma_{i}^{c} = \{ j : A_{i,j} = 1,~|y_i| \leq \epsilon \}$ and $\Gamma^c = \bigcup_{i\in \mathcal{E}}\Gamma^c_{i}$. The confined set $\Gamma$ is a complement of the set $\Gamma^c$, i.e., $\Gamma = [n]\backslash \Gamma^c$.
\end{definition}

The following Theorem gives the probability of $\{ \Omega \subseteq \Gamma \}$.

\begin{theorem}[Lower bound on $\mathbb{P} \{ \Omega \subseteq \Gamma \}$]\label{thm:conf}
  Suppose that in~(\ref{equ:linear_mod}), $\boldsymbol{A} \in \{0,1\}^{m \times n}$ is a random combinatorial matrix with independent columns, where each column is chosen uniformly among the vectors with $d$ ones. Furthermore, the $K$ nonzero components of $\boldsymbol{x}$ are independent and identically distributed~(i.i.d.), with the same cumulative distribution function~(CDF) $F_X(x) = \mathbb{P} \{X \leq x \}$ and probability density function~(PDF) $f_X(x)$. Then, the probability of $\{ \Omega \subseteq \Gamma \}$ is lower bounded by
\begin{equation}\label{equ:conf}
 \begin{split}
  &\mathbb{P} \{ \Omega \subseteq \Gamma \} \geq \\
    &1- m \sum_{\ell = 1}^{K} \binom{K}{\ell}  \left( \frac{d}{m} \right)^{\ell} \left( 1 - \frac{d}{m} \right)^{K - \ell} \Big( F_X^{*\ell}(\epsilon) - F_X^{*\ell}(-\epsilon) \Big),
 \end{split}
\end{equation}
where $F_X^{*\ell}(x) = (\underbrace{F_X * F_X * \cdots * F_X}_{\ell \text{ times}})(x)$ and the asterisk $*$ denotes the convolution operation.
\end{theorem}
\begin{proof}
See Appendix~\ref{sec:app0}.
\end{proof}

With Theorem~\ref{thm:conf}, we have the following two Corollaries.

\begin{corollary}\label{cor:conti}
If the PDF $f_X(x)$ of nonzero components of $\boldsymbol{x}$ is a continuous function, then the probability $\mathbb{P} \{ \Omega \subseteq \Gamma \} \to 1$ as $\epsilon \to 0$.
\end{corollary}
\begin{proof}
Obviously, $F_X^{*\ell}(x)$ is a continuous function provided that $f_X(x)$ is a continuous function. According to~Theorem~\ref{thm:conf}, we have $\underset{\epsilon \to 0}\lim \mathbb{P} \{ \Omega \subseteq \Gamma \} = 1$.
\end{proof}

\begin{corollary}\label{cor:discr}
If the values of nonzero components of $\boldsymbol{x}$ share the same polarity, then the probability $\mathbb{P} \{ \Omega \subseteq \Gamma \} \to 1$ as $\epsilon \to 0$.
\end{corollary}
The Corollaries indicate that by setting $\epsilon \to 0$, both Gaussian sparse signals and non-negative sparse signals can ensure $\mathbb{P} \{ \Omega \subseteq \Gamma   \} = 1$. We call these signals \emph{confined signals} in this paper. The definition of the confined signal is given as follows.
\begin{definition}[The confined signal]\label{def:conf_signal}
The confined signal is a $K$-sparse signal whose nonzero components are i.i.d. and, for $\ell = 1,2,\cdots,K$, their CDFs satisfy $F_X^{*\ell}(\epsilon)  - F_X^{*\ell}(-\epsilon) \to 0$ as $\epsilon \to 0$.
\end{definition}

\subsection{Confined OMP Algorithm}
The confined OMP is a modification of the OMP algorithm that introduces a confined set $\Gamma$. The key feature of the confined OMP algorithm is to introduce a confined set $\Gamma$ into the identification step such that the possibly chosen column indices are strictly confined to the confined set $\Gamma$. In the following, we assume that the signal $\boldsymbol{x}$ is a confined signal defined in~Definition~\ref{def:conf_signal} and has exactly $K$ nonzero elements. The details of the confined OMP algorithm are summarized in Algorithm~\ref{alg:COMP}. The necessary explanations of the confined OMP algorithm are as follows.
\begin{itemize}
  \item (\textbf{Preprocessing}) For a sufficiently small number $\epsilon$~(e.g., $\epsilon = 10^{-12}$), obtain the set $\mathcal{E} = \{i: |y_i| \leq \epsilon \}$ first. Then, obtain the confined set $\Gamma = [n]\backslash \bigcup_{i\in\mathcal{E}} \Gamma_i^{c}$, where $\Gamma_{i}^{c} = \{ j : A_{i,j} = 1,~|y_i| \leq \epsilon \}$. With $\Omega \subseteq \Gamma$, we have $\Omega = \Gamma$ if they have the same size. Thus, the identification is already done without the subsequent iterations.
  \item (\textbf{Identification}) If $|\Gamma| > K$, the iterative processing is executed. In each iteration, correlations between columns whose indices are in $\Gamma$ and the residual are compared. The column index corresponding to the maximal correlation is chosen as the new element of the estimated support set $\Lambda^{(k)}$.
\end{itemize}

\begin{algorithm}[t]
\caption{Confined orthogonal matching pursuit algorithm}
\label{alg:COMP}
\begin{algorithmic}
\STATE{Input: $\boldsymbol{y} \in \mathbb{R}^{m}$, $\boldsymbol{A} \in \{0,1\}^{m \times n}$, $\epsilon$, $\varepsilon$, and $K$. }
\STATE{Initialize: $k =0$, $\boldsymbol{r}^{(0)} = \boldsymbol{y}$, and $\Lambda^{(0)} = \emptyset$. }
\STATE{Preprocessing: $\mathcal{E} = \{i:|y_i| \leq \epsilon \}$ and $\Gamma = [n]\backslash \bigcup_{i\in\mathcal{E}} \Gamma_i^{c}$; }
\IF{$|\Gamma| = K$}
    \STATE  $\hat{\boldsymbol{x}}_{\Gamma} = \underset{\boldsymbol{x} \in \mathbb{R}^{|\Gamma|}}\argmin \vecnorm{\boldsymbol{y} - \boldsymbol{A}_{\Gamma} \boldsymbol{x} }_2$,
    \STATE \textbf{return} $\hat{\boldsymbol{x}}$.
\ENDIF
\WHILE{`$\vecnorm{\boldsymbol{r}^{(k)}}_2 > \varepsilon $ \textbf{and} $k < K$ are met' }
\STATE{$k=k+1$,}
\STATE{$t^{(k)}= \underset{i \in \Gamma \backslash \Lambda^{(k-1)} }\argmax |\boldsymbol{A}_i^{T}\boldsymbol{r}^{(k-1)} |$,~\textbf{(Identification)}}
\STATE{$\Lambda^{(k)} = \Lambda^{(k-1)} \cup \{t^{(k)} \}$,~\textbf{(Augmentation)}}
\STATE{$\hat{\boldsymbol{x}}_{\Lambda^{(k)}}^{(k)} = \underset{\boldsymbol{x} \in \mathbb{R}^{k}}\argmin \vecnorm{\boldsymbol{y} - \boldsymbol{A}_{\Lambda^{(k)}} \boldsymbol{x} }_2$,~\textbf{(Estimation)}}
\STATE{$\boldsymbol{r}^{(k)} = \boldsymbol{y} - \boldsymbol{A}_{\Lambda^{(k)}} \hat{\boldsymbol{x}}_{\Lambda^{(k)}}^{(k)}$.~\textbf{(Residual update)}}
\ENDWHILE
%\RETURN $T$
\RETURN{$\hat{\boldsymbol{x}}^{(K)}$. }
\end{algorithmic}
\end{algorithm}

We discuss the complexity of OMP algorithm and confined OMP algorithm as follows. It is known that the matrix-vector product can be divided into two steps: multiplication and addition. For a binary-valued matrix, the matrix-vector product only involves addition, no multiplication is required. If nonzero elements of $\boldsymbol{x}$ are real numbers, the addition requires at most $n(d-1)$ floating-point operations~(flops). Furthermore, the selection of the maximum inner-product value in the identification requires $n-1$ flops. Thus, the identification of OMP algorithm requires at most $Kn(d-1) + K(n-1)$ flops.

Thanks to $|\Gamma| \leq n$, the confined OMP algorithm requires fewer inner products in the identification step compared to the conventional OMP algorithm. Specifically, no inner product is required if $|\Gamma| = K$. Otherwise, the identification requires at most $K|\Gamma|(d-1)+ K(|\Gamma|-1)$ flops, where $K|\Gamma|(d-1)$ flops are required for the matrix-vector product operations and $K(|\Gamma|-1)$ flops for the selections of the maximum inner-product values. Furthermore, the preprocessing step involves additional operations. Specifically, there are extra $n$ flops to obtain $\mathcal{E}$ and the complexity of obtaining $\Gamma$ is negligible. Consider identification and preprocessing together, these steps require $n$ flops if $|\Gamma| = K$. Otherwise, at most $K|\Gamma|(d-1) + K(|\Gamma|-1) + n$ flops are required. Table~\ref{tab:complexity} summarizes the complexity of OMP algorithm and confined OMP algorithm in the identification step. The improvement of the proposed algorithm over the OMP algorithm is around $Kd$ times in terms of identification efficiency if $|\Gamma| = K$, and $\frac{Knd-K}{K|\Gamma|d-K+n}$ times otherwise.

\begin{table}[t]
\caption{Comparison of the complexity of the OMP algorithm and the confined OMP algorithm in the identification step.}\label{tab:complexity}
\centering
\begin{tabular}{|l|l|ll|}
\hline
Algorithm & \multicolumn{1}{c|}{OMP}  & \multicolumn{2}{l|}{confined OMP} \\ \hline
Condition  & \centering{ null}                     & \multicolumn{1}{l|}{$|\Gamma| = K$}   & $|\Gamma| > K$  \\ \hline
Flops      & $Knd - K$                        & \multicolumn{1}{l|}{$n$}     & $K|\Gamma|d-K+n$    \\ \hline
\end{tabular}
\end{table}

It can be foreseen that the confined OMP algorithm achieves a large reduction in complexity if $|\Gamma| \ll n$. Especially in case where $|\Gamma| = K$, the confined OMP algorithm eliminates the need for the identification. Even if $|\Gamma| = n$, the computational complexity of OMP algorithm and confined OMP algorithm are comparable, since the extra complexity introduced by the preprocessing step is negligible.

\subsection{Extension to the Generalized OMP Algorithm}
One advantage of the confined set is that it can be easily applied to many other sparse signal recovery algorithms. In this subsection, we use the generalized OMP~(gOMP) algorithm as an example for illustration.

The gOMP algorithm~\cite{Wang2012GOMP,Liu2012OMMP} is an efficient greedy recovery algorithm that allows selecting multiple column indices maximally correlated with the residual at each iteration, thereby reducing the required number of iterations. Similar to the confined OMP algorithm, the confined gOMP algorithm is a variant of gOMP algorithm that confines the support estimation to a reduced subset $\Gamma$ of the dictionary. Let $N$ be the number of column indices selected by the confined gOMP algorithm in the identification step of each iteration. The details of the confined gOMP algorithm are given in Algorithm~\ref{alg:GCOMP}. The main differences are summarized below.
\begin{itemize}
  \item (\textbf{Preprocessing}) If $|\Gamma| \leq \max\{K,N\}$, then the confined gOMP algorithm directly performs the least squares to estimate the sparse signal $\boldsymbol{x}$. Compared to the confined OMP algorithm, the confined gOMP algorithm allows for some extra redundancy in the preprocessing step when $K < N$.
  \item (\textbf{Identification}) In the identification step, if $N \leq |\Gamma \backslash \Lambda^{(k-1)}|$, column indices corresponding to the largest $N$ correlation in magnitude are selected to form a set $\mathcal{T}^{(k)}$. Otherwise, the remaining column indices in $\Gamma \backslash \Lambda^{(k-1)}$ are chosen.
  \item (\textbf{Stopping criteria}) The iterative process is terminated if $\vecnorm{\boldsymbol{r}^{(k)}}_2 \leq \varepsilon $, the iteration number reaches maximum $\min \{ K, \lfloor \frac{m}{N} \rfloor \}$, or the size of the estimated support set $\Lambda^{(k)}$ is equal to $|\Gamma|$.
\end{itemize}

\begin{algorithm}[t]
\caption{Confined generalized orthogonal matching pursuit algorithm}
\label{alg:GCOMP}
\begin{algorithmic}
\STATE{Input: $\boldsymbol{y} \in \mathbb{R}^{m}$, $\boldsymbol{A} \in \{0,1\}^{m \times n}$, $\epsilon$, $\varepsilon$, $N$, and $K$. }
\STATE{Initialize: $k =0$, $\boldsymbol{r}^{(0)} = \boldsymbol{y}$, and $\Lambda^{(0)} = \emptyset$. }
\STATE{Preprocessing: $\mathcal{E} = \{i:|y_i| \leq \epsilon \}$ and $\Gamma = [n]\backslash \bigcup_{i\in\mathcal{E}} \Gamma_i^{c}$; }
\IF{$|\Gamma| \leq \max\{K,N\}$}
    \STATE  $\hat{\boldsymbol{x}}_{\Gamma} = \underset{\boldsymbol{x} \in \mathbb{R}^{|\Gamma|}}\argmin \vecnorm{\boldsymbol{y} - \boldsymbol{A}_{\Gamma} \boldsymbol{x} }_2$,
    \STATE \textbf{return} $\hat{\boldsymbol{x}}$.
\ENDIF
\WHILE{`$\vecnorm{\boldsymbol{r}^{(k)}}_2 > \varepsilon $, $k < \min \{ K, \lfloor \frac{m}{N} \rfloor \}$, \textbf{and} $|\Lambda^{(k)}| < |\Gamma|$  are met' }
\STATE{$k=k+1$,}
\IF{$N \leq |\Gamma \backslash \Lambda^{(k-1)}|$}
    \STATE $\mathcal{T}^{(k)}= \underset{ \substack{ |\mathcal{T}^{(k)}| = N \\ i \in \Gamma \backslash \Lambda^{(k-1)} } }\argmax |\boldsymbol{A}_i^{T}\boldsymbol{r}^{(k-1)} |$,~\textbf{(Identification)}
\ELSE
    \STATE $\mathcal{T}^{(k)}= \Gamma \backslash \Lambda^{(k-1)}$,~\textbf{(Identification)}
\ENDIF
%\STATE{$\mathcal{T}^{(k)}= \underset{ \substack{ i \in \Gamma \backslash \Lambda^{(k-1)}  \\ |\mathcal{T}^{(k)}| = \min \{ N, |\Gamma \backslash \Lambda^{(k-1)}| \} } }\argmax |\boldsymbol{A}_i^{T}\boldsymbol{r}^{(k-1)} |$,}
\STATE{$\Lambda^{(k)} = \Lambda^{(k-1)} \cup \mathcal{T}^{(k)}$,~\textbf{(Augmentation)}}
\STATE{$\hat{\boldsymbol{x}}_{\Lambda^{(k)}}^{(k)} = \underset{\boldsymbol{x} \in \mathbb{R}^{|\Lambda^{(k)}|}}\argmin \vecnorm{\boldsymbol{y} - \boldsymbol{A}_{\Lambda^{(k)}} \boldsymbol{x} }_2$,~\textbf{(Estimation)}}
\STATE{$\boldsymbol{r}^{(k)} = \boldsymbol{y} - \boldsymbol{A}_{\Lambda^{(k)}} \hat{\boldsymbol{x}}_{\Lambda^{(k)}}^{(k)}$.~\textbf{(Residual update)}}
\ENDWHILE
%\RETURN $T$
\RETURN{$\hat{\boldsymbol{x}}^{(K)}$. }
\end{algorithmic}
\end{algorithm}

\section{Analysis of The Proposed Algorithm}\label{sec:analysis}
In this section, we first investigate the expectations of the sparsity of $\boldsymbol{y}$ and the size of the confined set $\Gamma$. Then, the recovery performance of the confined OMP algorithm using sparse random combinatorial matrices is analyzed under a noiseless linear system. Finally, we investigate the recovery performance and robustness of the proposed algorithm in a noisy linear system.

\subsection{The Sparsity of $\boldsymbol{y}$}
%Although $\boldsymbol{y}$ is known,
Studying the sparsity of $\boldsymbol{y}$ helps to derive the subsequent Lemmas and Theorems. Let $\epsilon \to 0$ and $\mathcal{E} = \{i: |y_i| \leq \epsilon \}$. Then, for $k\in\{1,2,\cdots,K\}$, let $\nu^{(k)} = m - |\mathcal{E}|$ be the number of ``nonzero" elements of $\boldsymbol{y}= \boldsymbol{A}\boldsymbol{x}$, where $\boldsymbol{x}$ is the confined signal and has $k$ nonzero elements. Obviously, the value of $\nu^{(k)}$ ranges from $d$ to $\min \{kd,m\}$. The following Lemma gives the probability of $\nu^{(k)}$ for $k=1,2,\cdots,K$.

\begin{lemma}\label{Lemm:pro_mu}
For any integer $K$, it holds that $\mathbb{P} \{ \nu^{(1)} = d \} = 1$ and
\begin{equation}\label{equ:pro_nu}
\mathbb{P} \{ \nu^{(k)} = \upsilon \} = \sum_{ z = \max \{ \upsilon - d ,d\}}^{\min \{\upsilon, (k-1)d \} } \frac{ \binom{z}{\upsilon-z} \binom{m-z}{d-\upsilon+z }}{ \binom{m}{d} } \mathbb{P} \{ \nu^{(k-1)} = z \}
\end{equation}
for $k = 2,3,\cdots, K$ and $\upsilon = d,d+1,\cdots,\min \{kd,m\}$, where $\mathbb{P} \{ \nu^{(k-1)} = z \}$ can be calculated by~(\ref{equ:pro_nu}) recursively.
\end{lemma}
\begin{proof}
See Appendix~\ref{sec:app1}.
\end{proof}

Once the probability $\mathbb{P} \{ \nu^{(K)} = \upsilon \}$ is obtained, one can easily obtain the expectation of $\nu^{(K)}$.

\begin{lemma}[The expectation of $\nu^{(K)}$]\label{thm:exp_y}
Suppose that in~(\ref{equ:linear_mod}), $\boldsymbol{A} \in \{0,1\}^{m \times n}$ is a random combinatorial matrix with $d$ ones per column and the signal $\boldsymbol{x}$ is a confined signal. Then,
\begin{equation}\label{equ:exp_y}
\mathbf{E}_\upsilon [ \nu^{(K)} ]= \sum_{\upsilon = d}^{\min\{Kd,m\}} \upsilon \cdot  \mathbb{P} \{ \nu^{(K)} = \upsilon \},
\end{equation}
where $\mathbb{P} \{ \nu^{(K)} = \upsilon \}$ is given in~(\ref{equ:pro_nu}). Also, it can be easily calculated by
\begin{equation}\label{equ:expy_simple}
\mathbf{E}_\upsilon [ \nu^{(K)} ]= m\left( 1 - \left(1 - \frac{d}{m}\right)^K \right).
\end{equation}
\end{lemma}

Fig.~\ref{fig:exp_y} shows empirical means and expectations of $\nu^{(K)}$ for different column degree $d$, where $\boldsymbol{A} \in \{0,1\}^{100\times 256}$ and $\boldsymbol{x}$ is a Gaussian sparse signal with exactly $K$ nonzero elements. It can be seen that the empirical results match well with their expectations.

\begin{figure}[t]
  \centering
  \includegraphics[width=7cm]{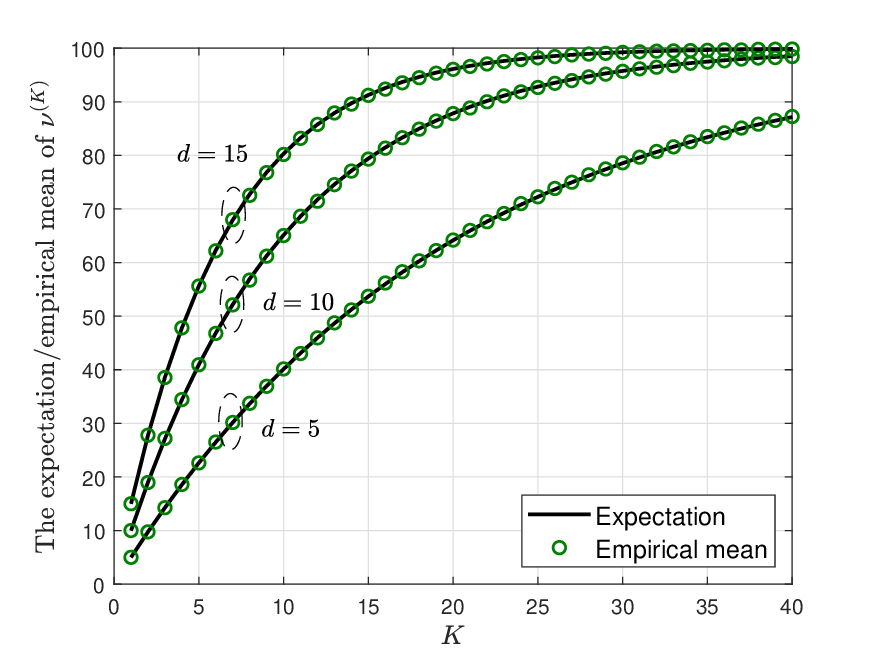}
  \caption{Empirical means and expectations of $\nu^{(K)}$ for different column degree $d$.}
  \label{fig:exp_y}
\end{figure}

\subsection{The Size of $\Gamma$}
To clarify the effectiveness of removing the redundancy of $\boldsymbol{A}$, it is necessary to consider the size of $\Gamma$ theoretically. Furthermore, the size of $\Gamma$ is crucial to the complexity and recovery performance analysis of the proposed algorithm. The following Theorem will give $\mathbf{E}[ |\Gamma| ]$.

\begin{theorem}[The expectation of $|\Gamma|$]\label{thm:exp_Gamma}
Suppose that in~(\ref{equ:linear_mod}), $\boldsymbol{A} \in \{0,1\}^{m \times n}$ is a random combinatorial matrix with $d$ ones per column and the signal $\boldsymbol{x}$ is a confined signal. Then,
\begin{equation}\label{equ:exp_Gamma}
\mathbf{E}[ |\Gamma| ]= K + (n-K) \cdot \mathbf{E}_{\upsilon} \left[ \frac{\binom{\nu^{(K)}}{d}}{\binom{m}{d}} \right],
\end{equation}
where $\mathbf{E}_{\upsilon} [ \cdot ]$ is given in~(\ref{equ:exp_y}).
\end{theorem}
\begin{proof}
See~Appendix~\ref{sec:app2}.
\end{proof}

\begin{figure}[t]
  \centering
  \includegraphics[width=7cm]{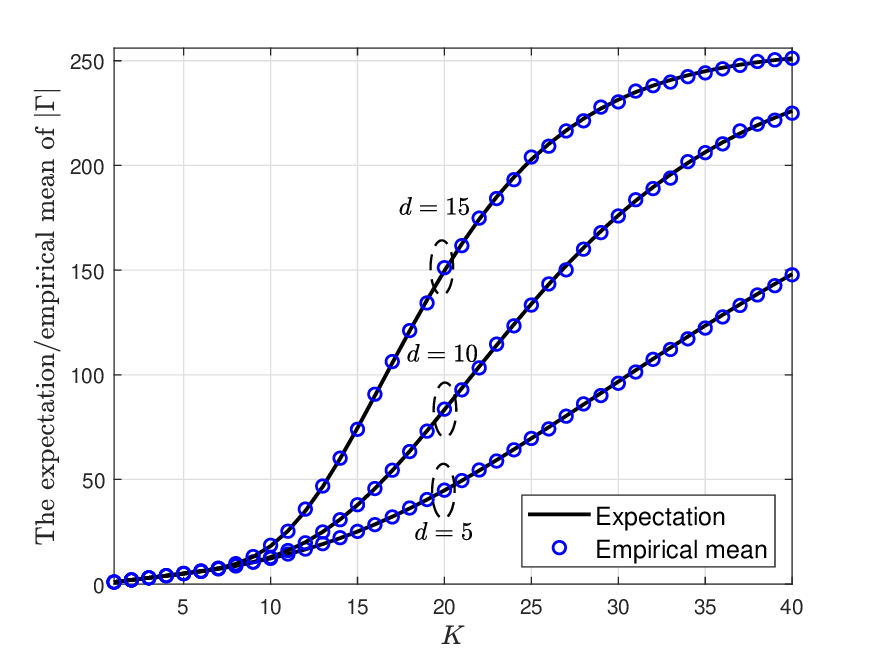}
  \caption{Empirical means and expectations of $|\Gamma|$ for different column degree $d$.}
  \label{fig:exp_gamma}
\end{figure}

Fig.~\ref{fig:exp_gamma} shows empirical means and expectations of $|\Gamma|$ for different column degree $d$, where $\boldsymbol{A} \in \{0,1\}^{100\times 256}$ and $\boldsymbol{x}$ is a Gaussian sparse signal with exactly $K$ nonzero elements. It can be seen that the empirical results match well with their expectations. We also observe that $\mathbf{E}[ |\Gamma| ] \approx K$ when $K$ is relatively small. These observations confirm the effectiveness of removing the redundancy of $\boldsymbol{A}$.

\subsection{Analysis of Confined OMP in the Noiseless Linear System}
We define $spark(\boldsymbol{A})$ as the spark of $\boldsymbol{A}$, where $spark(\boldsymbol{A}) := \min \{ \vecnorm{\boldsymbol{x}}_0 :  \boldsymbol{A}\boldsymbol{x} = \boldsymbol{0}, \boldsymbol{x} \neq  \boldsymbol{0} \}$. The condition $spark(\boldsymbol{A}) > K$ is essential to guarantee that the least squares problem in the proposed algorithm has a unique solution. The spark of $\boldsymbol{A}$ depends heavily on the column degree $d$. For example, when the column degree is set to $d = 1$, we have $spark(\boldsymbol{A}) = 2$. However, the computation of $spark(\boldsymbol{A})$ is NP-hard in general. Fortunately, the following Theorem and Corollary show that when $m$ is sufficiently large, $\boldsymbol{A}$ meets $spark(\boldsymbol{A}) > K$ with high probability.

\begin{theorem}[Theorem~1.2 in~\cite{Ferber2022Singularity}]\label{thm:pro_fullrank}
Fix $\gamma > 0$, and let $d=d(m)$ be any function of $m$ satisfying $\min\{d,m-d\} \geq (1+\gamma)\ln m$. Then, for an $m \times m$ random combinatorial matrix $\boldsymbol{Q}$ with independent rows, where each row is chosen uniformly among the vectors with $d$ ones, we have
\begin{displaymath}
\lim_{m \to \infty} \mathbb{P} \{ \boldsymbol{Q}~\mathrm{is}~\mathrm{singular} \} \to 0.
\end{displaymath}
\end{theorem}

\begin{corollary}\label{cor:pro_fullrank}
Fix $\gamma > 0$, and let $d=d(m)$ be any function of $m$ satisfying $ (1+\gamma)\ln m \leq d \leq m/2$. Then, for an $m \times n$ random combinatorial matrix $\boldsymbol{A}$ with $d$ ones per column and a sufficiently large $m$, we have
\end{corollary}
\begin{equation}\label{equ:pro_rank}
\mathbb{P} \left\{ spark(\boldsymbol{A}) > K  \right\} = 1 - o(1).
\end{equation}
\begin{proof}
According to~Theorem~\ref{thm:pro_fullrank}, we know that the $m \times m$ random combinatorial matrix $\boldsymbol{Q}$ is nonsingular with probability $1-o(1)$ for a sufficiently large $m$. In other words, $m$ columns of $\boldsymbol{Q}^T$ are linearly independent with probability $1-o(1)$ for a sufficiently large $m$. For the $m \times n$ random combinatorial matrix $\boldsymbol{A}$ with a sufficiently large $m$, any $K$ columns of $\boldsymbol{A}$ are obviously linearly independent with probability $1-o(1)$, resulting in $spark(\boldsymbol{A}) > K$ with probability $1 - o(1)$. Thus, we get~(\ref{equ:pro_rank}).
\end{proof}

In general, the sparsity level $K$ is much smaller than $m$. Thus, as $m$ increases, the probability $ \mathbb{P} \{ spark(\boldsymbol{A}) > K \}$ converges to 1 more rapidly than the probability $\mathbb{P} \{  \boldsymbol{Q}~\mathrm{is}~\mathrm{nonsingular}\}$ shown in~Theorem~\ref{thm:pro_fullrank}. In summary, when $m$ is sufficiently large and the column degree $d$ satisfies $ (1+\gamma)\ln m \leq d \leq m/2$ for a constant $\gamma > 0$, the spark of $\boldsymbol{A}$ is greater than $K$ with probability $1-o(1)$. In the following theoretical analysis, we will take this condition as a premise.

Define the event $\mathbb{S}_{\mathrm{comp}}=\{$confined OMP exactly recovers the $K$-sparse confined signal $\boldsymbol{x}\}$. In the following, we will present a lower bound on the probability that the confined OMP algorithm exactly recovers the confined signal $\boldsymbol{x}$ for $\boldsymbol{A} \in \{0,1\}^{m\times n}$.

\begin{theorem}[Recovery probability of confined OMP]\label{thm:lowerbound_comp}
Suppose that in~(\ref{equ:linear_mod}), $\boldsymbol{A}$ is an $m \times n$ random combinatorial matrix with $d$ ones per column, and $\boldsymbol{x}$ is a confined signal with exactly $K$ nonzero entries. If $\boldsymbol{A}$ satisfies $spark(\boldsymbol{A}) > K$, then it holds that
\begin{equation}\label{equ:lowerbound_COMP}
\mathbb{P} \{ \mathbb{S}_{\mathrm{comp}} \} \geq \mathbf{E}_{\upsilon} \left[  \left(1 - \frac{\binom{\nu^{(K)}}{d}}{\binom{m}{d}} \right)^{n-K}  \right],
\end{equation}
where $\mathbf{E}_{\upsilon} [ \cdot ]$ is given in~(\ref{equ:exp_y}).
\end{theorem}
\begin{proof}
See~Appendix~\ref{sec:app5}.
\end{proof}

Based on the theoretical results presented in Theorem~\ref{thm:lowerbound_comp}, the following Corollary will clearly reveal the relationships between the various parameters.

\begin{corollary}\label{cor:lower_boundlooser}
Given fixed $m$ and $d$, if $K \leq m/d$, then it holds that
\begin{equation}\label{equ:lower_boundlooser}
\mathbb{P} \{ \mathbb{S}_{\mathrm{comp}} \} \geq \bar{\pi}_K = \left( 1 - \left( \frac{Kd}{m} \right)^d \right)^{n-K}.
\end{equation}
Furthermore, given fixed $m$ and $K$, the lower bound $\bar{\pi}_K$ reaches its maximum value
\begin{equation}
  \bar{\pi}_K^{*} = \left( 1 -   \exp \left(-\frac{m}{K\cdot \e} \right) \right)^{n-K}
\end{equation}
when $d = \frac{m}{K \cdot \e}$.
\end{corollary}
\begin{proof}
See~Appendix~\ref{sec:app6}.
\end{proof}

{\em Remark~1:} In Corollary~\ref{cor:lower_boundlooser}, $d = \frac{m}{K\cdot \e}$ may not be an integer. Actually, the maximum value of $\bar{\pi}_K$ is achieved when $d$ is set to either $\lfloor \frac{m}{K\cdot \e} \rfloor$ or $ \lceil \frac{m}{K \cdot \e} \rceil$. To enhance readability and more clearly analyze the impact of different parameters on $\bar{\pi}_K$, we allow $d$ to take real values in Corollary~\ref{cor:lower_boundlooser}. Similarly, the same applies to the following theoretical analysis.

Corollary~\ref{cor:lower_boundlooser} shows that a higher probability $\bar{\pi}_K$ can be achieved by increasing measurements $m$ or decreasing the signal dimension $n$. Furthermore, the choice of $d$ impacts the probability $\bar{\pi}_K$ for the different sparse level $K$. While choosing a small $d$ can enhance $\bar{\pi}_K$ for larger $K$, it may also cause $\bar{\pi}_K$ for smaller $K$ to deviate from their maximum values.

\begin{theorem}[Necessary number of measurements]\label{thm:nec_m}
Suppose that in~(\ref{equ:linear_mod}), $\boldsymbol{A}$ is an $m \times n$ random combinatorial matrix with $d$ ones per column and $\boldsymbol{x}$ is a confined signal with exactly $K$ nonzero entries. Given a real number $\beta > 1$ and $K$, if $m = c\frac{K}{\ln \beta} \ln (n - K) < n$ for a constant $c > \beta$ and $d = \frac{m}{K \beta}$ can ensure that $\boldsymbol{A}$ satisfies $spark(\boldsymbol{A}) > K$, then it holds that
\begin{equation}\label{equ:nec_m}
  \mathbb{P} \{ \mathbb{S}_{\mathrm{comp}} \} \geq 1 - (n-K)^{1-c/\beta}.
\end{equation}
\end{theorem}
\begin{proof}
See~Appendix~\ref{sec:app7}.
\end{proof}

In Theorem~\ref{thm:nec_m}, (\ref{equ:nec_m}) is simplified to
\begin{equation}
  \mathbb{P} \{ \mathbb{S}_{\mathrm{comp}} \} \geq 1 - \frac{1}{n-K}
\end{equation}
if $m = 2\frac{\beta}{\ln \beta} K \ln (n - K) < n$. Let $g(\beta) = 2\frac{\beta}{\ln \beta} K \ln (n - K)$. A simple derivation can yield that the function $g(\beta)$ reaches its minimum value when $\beta = \e$. Thus, if $\boldsymbol{A}$ satisfies $spark(\boldsymbol{A}) > K$, the minimum measurements $m = 2\e K \ln (n - K)$ are sufficient to guarantee that the recovery probability $\mathbb{P} \{ \mathbb{S}_{\mathrm{comp}} \}$ of the proposed algorithm is no lower than $1 - \frac{1}{n-K}$.

\begin{corollary}\label{cor:asymptotic}
Let the signal dimension $n = m^\tau$ with $\tau > 1$ and the column degree $d = \gamma\ln m$ with $\gamma > 1$. In the asymptotic regime in which both $m$ and $n$ tend to infinity, if the sparsity level $K < \frac{1}{\gamma} \cdot \e^{-\tau/\gamma}  \cdot \frac{m}{\ln m}$, then it holds that
\begin{equation}
  \lim_{\substack{m \to \infty \\ n \to  \infty}} \mathbb{P} \{ \mathbb{S}_{\mathrm{comp}} \} = 1.
\end{equation}
Furthermore, the sparsity level $K$ reaches its maximum value $\left\lfloor \frac{1}{\tau \cdot \e} \frac{m}{\ln m} \right\rfloor$ when $\gamma = \tau$.
\end{corollary}
\begin{proof}
See Appendix~\ref{sec:app10}
\end{proof}

Corollary~\ref{cor:asymptotic} shows that when both $m$ and $n$ tend to infinity as $n = m^\tau$, any confined signal with the sparsity level $K = o\left(\frac{m}{\ln m}\right)$ can be exactly recovered with probability 1. In order to ensure that $spark(\boldsymbol{A}) > K$ with probability $1 - o(1)$ for a sufficiently large $m$, as shown in~Corollary~\ref{cor:pro_fullrank}, the column degree should satisfy $d = \gamma \ln m$ where $\gamma >1$. Moreover, $d = \tau \ln m$ is the optimal choice for selecting the column degree in the asymptotic regime.

\subsection{Analysis of Confined OMP in the Noisy Linear System}\label{sec:noisy_mod}
In practical applications, measurement noise is non-negligible. Thus, in this subsection, we investigate the robustness of the proposed algorithm to noise in a noisy linear system:
\begin{equation}\label{equ:noise_linear_mod}
\boldsymbol{y}_{\mathrm{eff}}  = \boldsymbol{A} \boldsymbol{x} + \boldsymbol{v},
\end{equation}
where $\boldsymbol{v} \in \mathbb{R}^m$ is a noise vector satisfying $\vecnorm{\boldsymbol{v}}_\infty \leq \eta$ for a real number $\eta > 0$. The following Theorem will give the probability $ \mathbb{P} \{ \Omega \subseteq \Gamma \} $ in a noisy linear system.
\begin{theorem}[Lower bound on $\mathbb{P} \{ \Omega \subseteq \Gamma \}$ with noisy $\boldsymbol{y}_{\mathrm{eff}}$]\label{thm:noise_confset}
Suppose that in (\ref{equ:noise_linear_mod}), $\boldsymbol{A}$ is an $m \times n$ random combinatorial matrix with $d$ ones per column, $\boldsymbol{x}$ is a signal with exactly $K$ nonzero components, and $\boldsymbol{v}$ is a noise vector with $\vecnorm{\boldsymbol{v}}_{\infty} \leq \eta$. Furthermore, these nonzero components of $\boldsymbol{x}$ are i.i.d., with the same CDF $F_X(x)$ and PDF $f_X(x)$. Then, for $\ell \in [K]$, the probability of $\{ \Omega \subseteq \Gamma \}$ is lower bounded by
\begin{equation}\label{equ:conf_noise}
  \mathbb{P} \{ \Omega \subseteq \Gamma \} \geq  1- \mathbf{E}_\upsilon [ \nu^{(K)} ] \cdot \underset{\ell \in [K]}\max \left( F_X^{*\ell}(\epsilon + \eta) - F_X^{*\ell}(-\epsilon - \eta) \right),
\end{equation}
where $\mathbf{E}_\upsilon [ \nu^{(K)} ]$ is given in~(\ref{equ:expy_simple}).
\end{theorem}
\begin{proof}
See~Appendix~\ref{sec:app8}.
\end{proof}

In contrast to the noiseless linear system with $\epsilon \to 0$, the parameter $\epsilon$ is crucial in the noisy linear system. An excessively large $\epsilon$ may shrink the size of $\Gamma$ and increase the risk of missing columns in the support set $\Omega$. Conversely, as $\epsilon \to 0$, a large number of redundant column indices may emerge in $\Gamma$, e.g., $|\Gamma| = n$. To maximize the elimination of redundancy and facilitate the subsequent analysis, the value of $\epsilon$ should be set to $\eta$ such that noise-induced perturbations remain within the tolerance range $\epsilon = \eta$, resulting in a new set $\mathcal{E}_{\mathrm{eff}} = \{i: |y_i| \leq \eta\}$.

Let $\nu^{(K)}_{\mathrm{eff}} = m - |\mathcal{E}_{\mathrm{eff}}|$ be the number of ``nonzero" entries of $\boldsymbol{y}_{\mathrm{eff}}$ and define $\mathbb{S}_{\mathrm{comp}}^{(\mathrm{supp})}=\{$confined OMP exactly recovers the support of the $K$-sparse confined signal $\boldsymbol{x}$ from a noisy measurements $\boldsymbol{y}_{\mathrm{eff}}$ $\}$. The following Corollary will give the lower bound on $\mathbb{P} \{ \mathbb{S}_{\mathrm{comp}}^{(\mathrm{supp})} \}$.

\begin{corollary}\label{cor:pi_Ksupp}
If $\boldsymbol{A}$ satisfies $spark(\boldsymbol{A}) > K$, then it holds that
\begin{equation}\label{equ:pi_Ksupp}
 \begin{split}
\mathbb{P} \{ \mathbb{S}_{\mathrm{comp}}^{(\mathrm{supp})} \}& \geq \left( 1 - \left( \frac{Kd}{m} \right)^d \right)^{n-K}\cdot\\
&\left( 1- \mathbf{E}_\upsilon [ \nu^{(K)} ] \cdot \underset{\ell \geq 1}\max \left( F_X^{*\ell}(2\eta) - F_X^{*\ell}(-2\eta) \right) \right),
 \end{split}
\end{equation}
where $\mathbf{E}_{\upsilon} [ \nu^{(K)} ]$ is given in~(\ref{equ:expy_simple}).
\end{corollary}
\begin{proof}
The probability $\mathbb{P} \{ \mathbb{S}_{\mathrm{comp}}^{(\mathrm{supp})} \}$ is lower bounded by
\begin{equation}
  \mathbb{P} \{ \mathbb{S}_{\mathrm{comp}}^{(\mathrm{supp})} \} \geq \mathbb{P} \{ |\Gamma| = K | \Omega \subseteq \Gamma \} \cdot \mathbb{P} \{ \Omega \subseteq \Gamma \},
\end{equation}
where the lower bound on $\mathbb{P} \{ \Omega \subseteq \Gamma \}$ is given in~(\ref{equ:conf_noise}). Since $\epsilon$ is set to $\eta$, all indices of ``zero" entries and some indices of ``nonzero" of $\boldsymbol{y}_{\mathrm{eff}}$ are contained in the set $\mathcal{E}_{\mathrm{eff}}$, resulting in $\nu^{(K)}_{\mathrm{eff}} \leq Kd $. Thus, the lower bound on $\mathbb{P} \{ |\Gamma| = K | \Omega \subseteq \Gamma \}$ can be obtained by~(\ref{equ:lower_boundlooser}). Finally, we get~(\ref{equ:pi_Ksupp}).
\end{proof}

It can be seen from~(\ref{equ:pi_Ksupp}) that the robustness of the proposed algorithm to noise is significantly influenced by the probabilistic distribution of the nonzero components of the signal. In the following, we use Gaussian sparse signals as a case in point for illustration.

Assume that the $K$ nonzero components of $\boldsymbol{x}$, denoted by $X_1,X_2,\cdots,X_K$, are i.i.d., with the same Gaussian distribution $\mathcal{N}(\mu , \sigma^2)$. Then, for $\ell = 1,2,\cdots,K$, we have
\begin{equation}
  F_X^{*\ell}(2\eta) - F_X^{*\ell}(-2\eta) = \Phi \left( \frac{2\eta  - \ell \mu}{\sigma \sqrt{\ell}}  \right) - \Phi \left( \frac{- 2\eta  - \ell \mu}{\sigma \sqrt{\ell}}  \right),
\end{equation}
where $\Phi(\cdot)$ is the CDF of the standard normal distribution. According to Corollary~\ref{cor:pi_Ksupp}, the probability $\mathbb{P} \{ \mathbb{S}_{\mathrm{comp}}^{(\mathrm{supp})} \}$ is lower bounded by
\begin{equation}\label{equ:piK_gaussian}
 \begin{split}
\mathbb{P} &\{ \mathbb{S}_{\mathrm{comp}}^{(\mathrm{supp})} \} \geq \left( 1 - \left( \frac{Kd}{m} \right)^d \right)^{n-K}  \cdot\\
&\left( 1- \mathbf{E}_\upsilon [ \nu^{(K)} ] \left( \Phi \left( \frac{2\eta  - \mu}{\sigma}  \right) - \Phi \left( \frac{- 2\eta  - \mu}{\sigma}  \right) \right) \right).
 \end{split}
\end{equation}
Fixed $\eta$ and $\sigma$, increasing $\mu$ can reduce the value of $\Phi \left( \frac{2\eta  - \mu}{\sigma}  \right) - \Phi \left( \frac{- 2\eta  - \mu}{\sigma}  \right)$, as the Gaussian distribution shifts rightward. This makes the second term of~(\ref{equ:piK_gaussian}) approach 1, thereby improving the lower bound on $\mathbb{P} \{ \mathbb{S}_{\mathrm{comp}}^{(\mathrm{supp})} \}$. Specifically, when $\mu$ becomes sufficiently large compared to the noise threshold $\eta$, the probability of noise corrupting the support recovery diminishes exponentially.

\section{Experimental Results}\label{sec:experiments}
This section presents the experimental results in both noiseless and noisy cases, demonstrating the advantage of the confined OMP algorithm in terms of recovery performance and complexity. Furthermore, we use the lower bound on the probability $\mathbb{P} \{ \mathbb{S}_{\mathrm{comp}} \}$ to optimize the column degree $d$. The experimental results were obtained by MATLAB R2023a on a desktop computer with Intel(R) Core(TM) i7-11700 CPU @ 2.50 GHz. In the following simulations, we generate a $K$-sparse signal with exactly $K$ nonzero elements, whose support is chosen at random. In addition, we consider two types of $K$-sparse signals: (a)~Gaussian sparse signals and (b)~flat sparse signals. The nonzero elements of the Gaussian sparse signal are independently and randomly drawn from a Gaussian distribution $\mathcal{N}(\mu,1)$. Furthermore, the support of the flat sparse signal is randomly chosen and the nonzero elements are set to $\theta$, where $\theta$ is a real number. We conduct an experiment using $1000$ Monte-Carlo trials and $\boldsymbol{A}$ is generated randomly in each trial. The relative recovery error is defined as
\begin{equation}
  \frac{\vecnorm{\boldsymbol{x} - \hat{\boldsymbol{x}} }_2}{ \vecnorm{\boldsymbol{x}}_2 },
\end{equation}
where $\hat{\boldsymbol{x}}$ is the recovered signal. If the relative recovery error is less than or equal to $0.001$, we declare this recovery to be perfect.

\begin{figure*}[t]
	\centering
	\subfigure[]{
		\includegraphics[width=7cm]{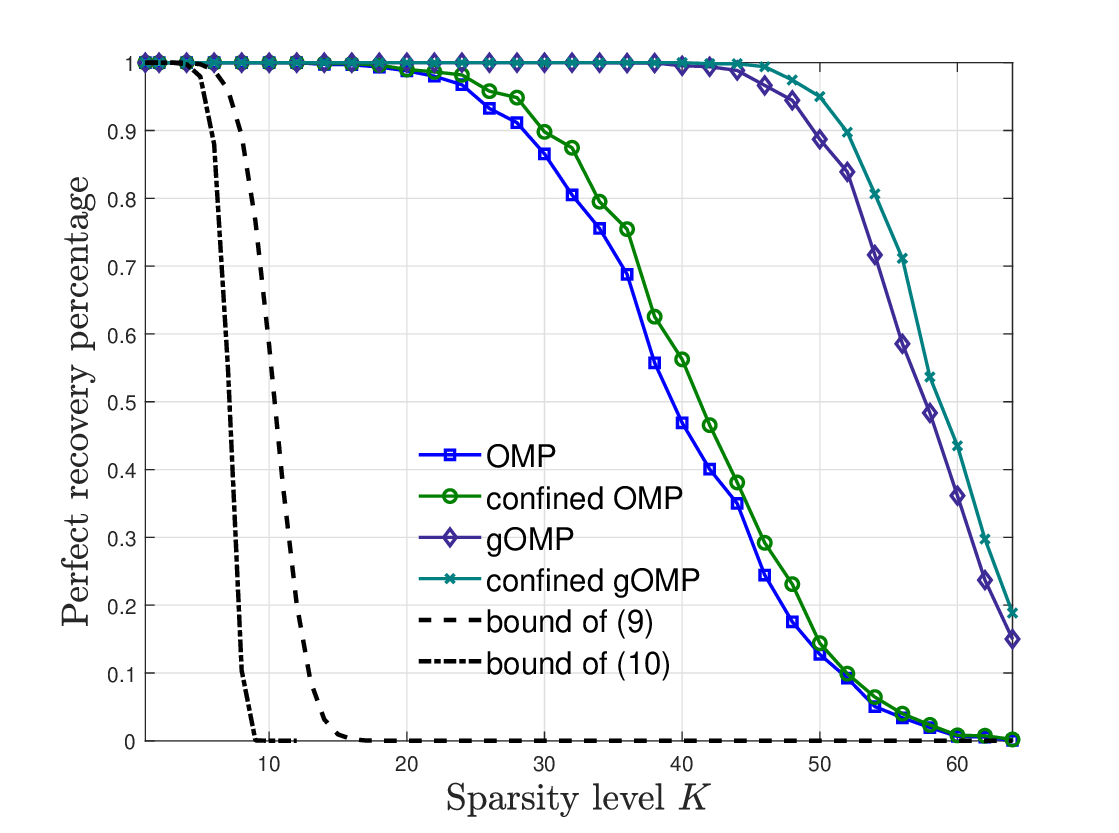}
	}
	\subfigure[]{
		\includegraphics[width=7cm]{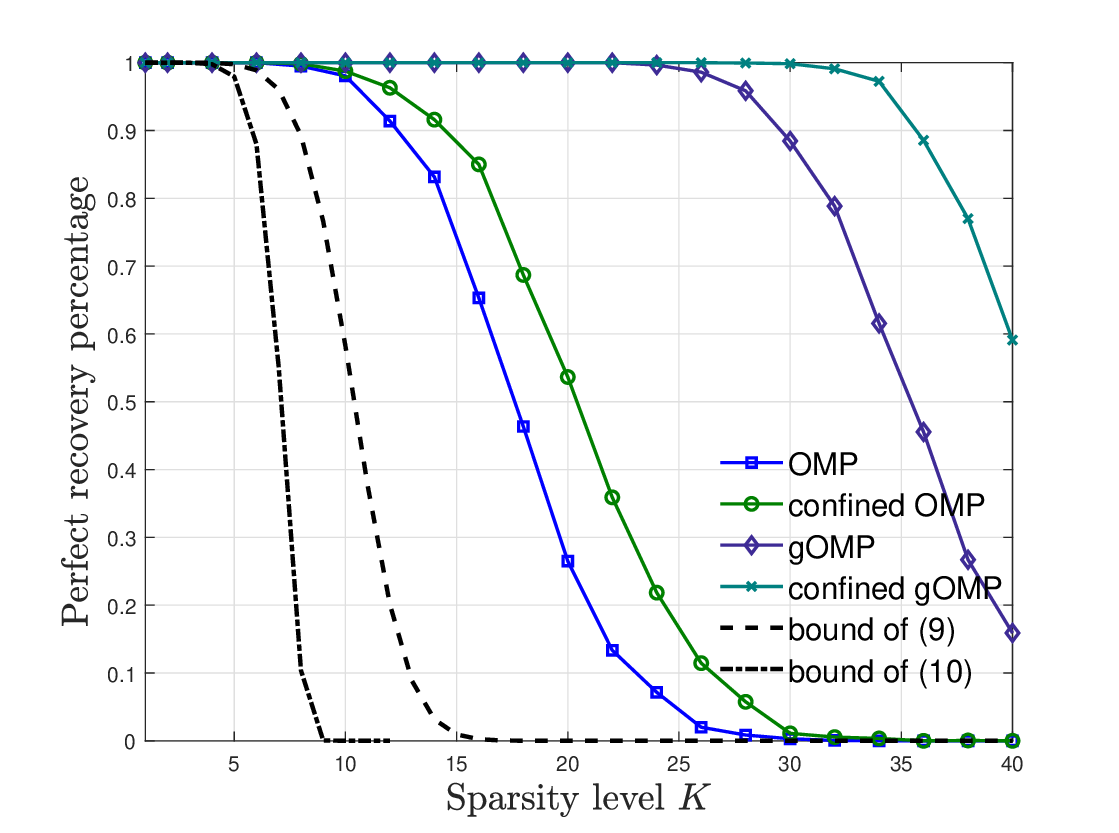}
	}
	\caption{Perfect recovery percentage for $K$-sparse (a)~Gaussian sparse signals and (b)~flat sparse signals versus the sparsity $K$.  }\label{fig:Sim_1}
\end{figure*}

\begin{figure*}[t]\label{fig:Sim_2}
	\centering
	\subfigure[]{
		\includegraphics[width=7cm]{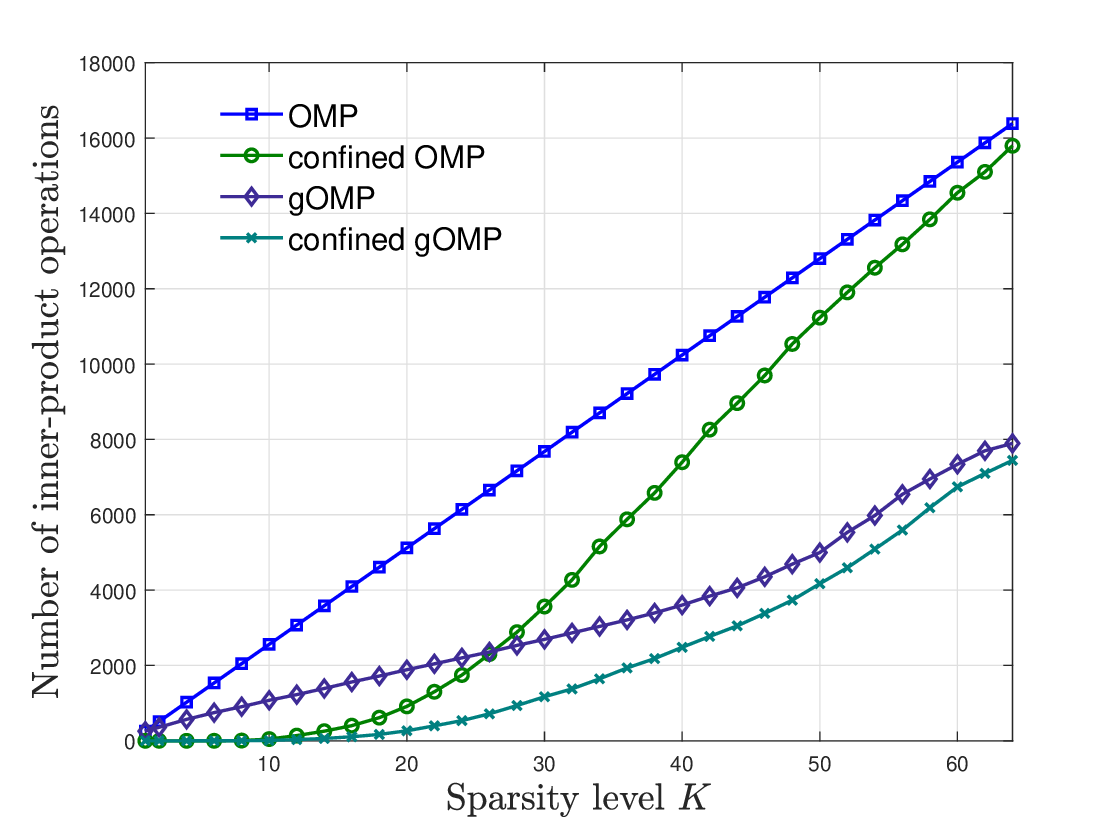}
	}
	\subfigure[]{
		\includegraphics[width=7cm]{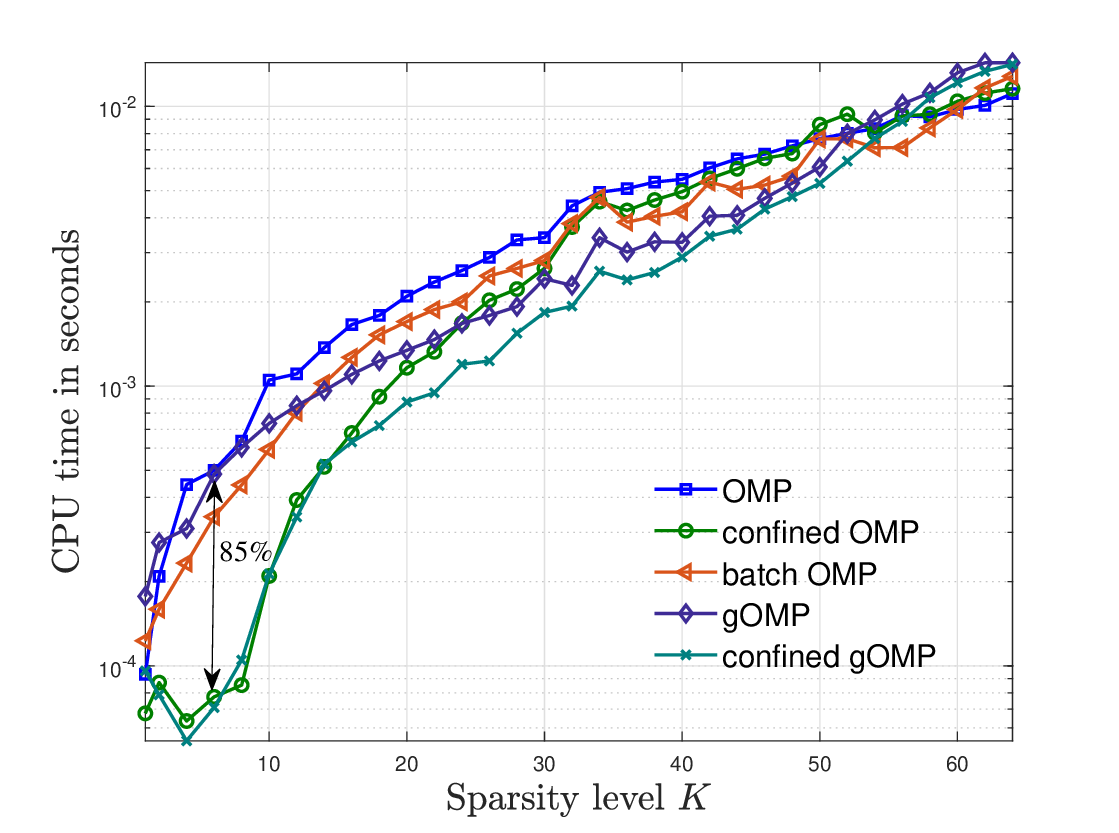}
	}
	\caption{(a)~Average number of inner-product operations in the identification step and (b) average CPU times of different greedy algorithms for Gaussian sparse signals with different $K$.  }\label{fig:Sim_2}
\end{figure*}

\subsection{Recovery Performances for Different Sparsity $K$}
We first show the efficiency and recovery performance of the proposed confined OMP and confined gOMP algorithms for different sparsity $K$ in the noiseless case. In Fig.~\ref{fig:Sim_1} and Fig.~\ref{fig:Sim_2}, each column of the measurement matrix $\boldsymbol{A}\in \{0,1\}^{128 \times 256}$ is chosen independently and uniformly among the vectors with $10$ ones. The setting $N = 3$ is applied to both the gOMP and confined gOMP algorithms and a sufficiently small $\epsilon = 10^{-12}$ is set for both the confined OMP and confined gOMP algorithms. The iterative process for all greedy algorithms terminates if the residual $\vecnorm {\boldsymbol{r}^{(k)}}_2 \leq 10^{-5}$. The signal parameters $\mu = 0$ and $\theta = 1$ are set for Gaussian sparse signals and flat sparse signals, respectively.

In Fig.~\ref{fig:Sim_1}, we can see that both the confined OMP and confined gOMP algorithms achieve performance gains for both Gaussian and flat sparse signals. In particular, the performance gain for the flat sparse signal is appreciable. It is reasonable because when the sparsity $K$ is relatively small, the size of $\Gamma$ is not as large, which enables the confined greedy algorithms~(including the confined OMP and confined gOMP algorithms) to screen out more interference.

Fig.~\ref{fig:Sim_2} shows the average inner-product operations in the identification step and average CPU times of different greedy algorithms for Gaussian sparse signals with different sparsity $K$. Note that the experimental results in Fig.~\ref{fig:Sim_2} and Fig.~\ref{fig:Sim_1}~(a) are obtained from the same experimental run. From~Fig.~\ref{fig:Sim_2}(a), we observe that the number of inner-product operations of OMP algorithm increases linearly with the increase of $K$. The reason is that the OMP algorithm requires $Kn$ inner-product operations in the identification step to recover a $K$-sparse signal. For the confined greedy algorithm, however, the identification is eliminated if $|\Gamma| = K$. Otherwise, at most $K|\Gamma|$ inner-product operations are required to recover a $K$-sparse signal. As shown in~Fig.~\ref{fig:exp_gamma}, the expectation of $|\Gamma|$ is much smaller than $n$ for a relatively small $K$. Thus, in~Fig.~\ref{fig:Sim_2}(a), the confined OMP algorithm achieves a significant reduction in number of inner-product operations for a relatively small $K$. The average CPU time shown in~Fig.~\ref{fig:Sim_2}(b) relates to the number of inner-product operations. Their trends of curves shown in~Fig.~\ref{fig:Sim_2}(a) and (b) are almost the same. In particular, the confined OMP algorithm achieves a reduction of about $85\%$ in average CPU time to recover a $6$-sparse Gaussian sparse signal when compared to that of OMP algorithm. Furthermore, we also present the average CPU time of the batch OMP algorithm proposed in~\cite{Rubinstein2008Efficient} for comparison. It can be seen that the complexity reduction of the confined OMP algorithm is much greater than that of the batch OMP algorithm for a relatively small $K$. For a relatively large sparsity level $K$, however, the average CPU times of the confined greedy algorithms are almost identical to those of their benchmarks, particularly when $K \geq 32$. This is reasonable because the complexity of solving the least squares problem correlates with the sparsity level $K$. Specifically, the complexity of solving the least squares problem in the OMP algorithm is $\mathcal{O}(mK^2)$. When $K$ is relatively large, the complexity of solving the least squares problem dominates.

\emph{Remark~3:} The batch OMP algorithm is an efficient implementation of the OMP algorithm. Their recovery performances are the same. Thus, we omit the performance comparison between the batch OMP algorithm and the proposed algorithm. More details about the batch OMP algorithm can be found in~\cite{Rubinstein2008Efficient}.

\begin{figure*}[t]
	\centering
	\subfigure[]{
		\includegraphics[width=7cm]{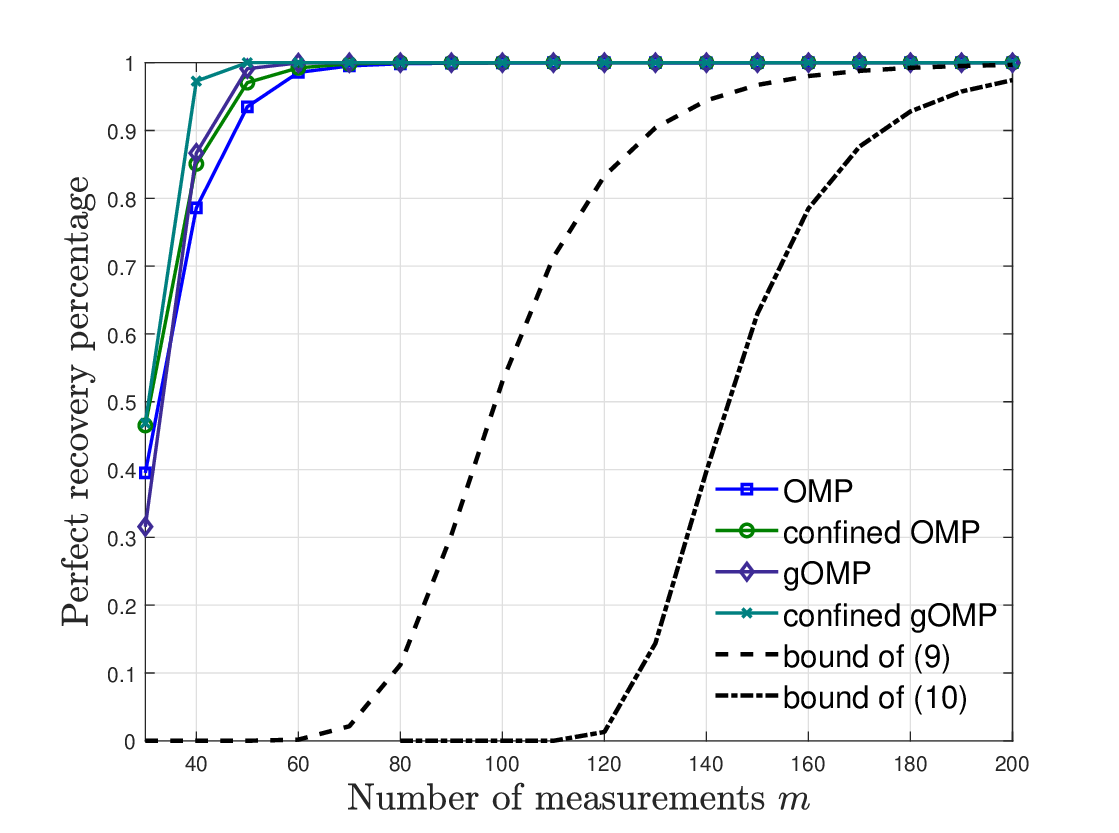}
	}
	\subfigure[]{
		\includegraphics[width=7cm]{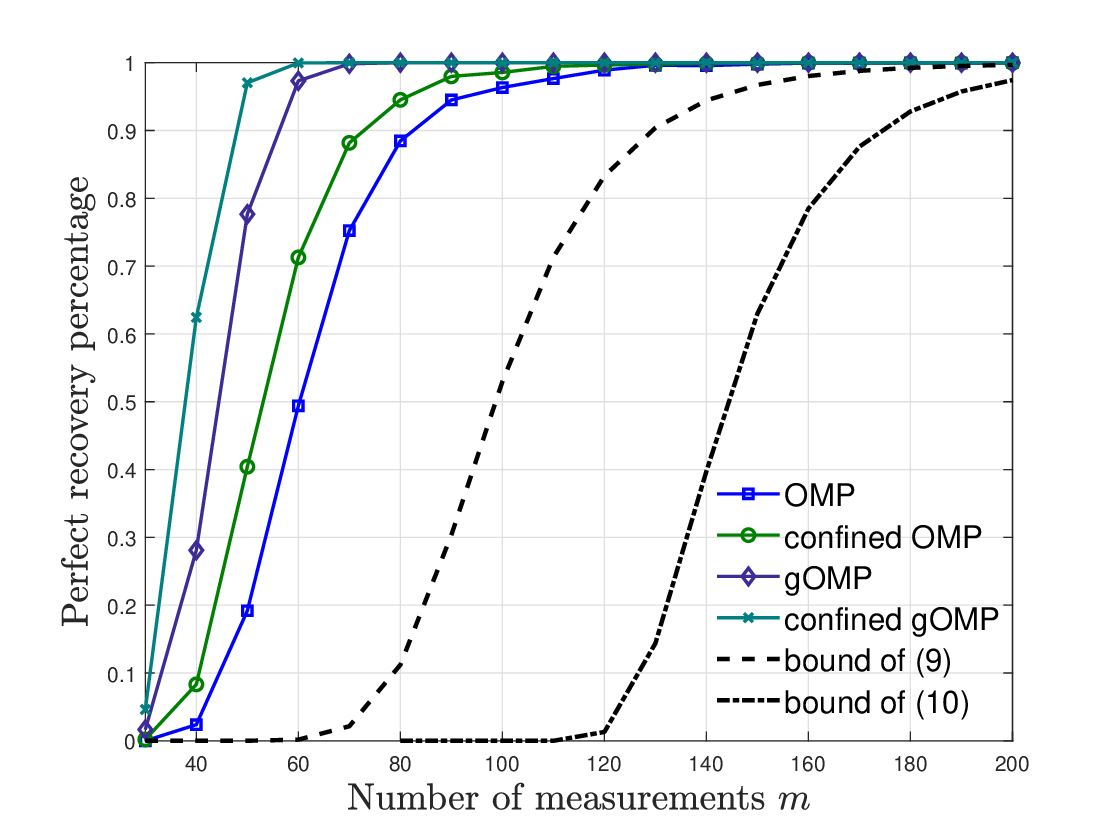}
	}
	\caption{Perfect recovery percentage for $8$-sparse (a)~Gaussian sparse signals and (b)~flat sparse signals versus measurements $m$.  }\label{fig:Sim_3}
\end{figure*}
\begin{figure*}[t]
	\centering
	\subfigure[]{
		\includegraphics[width=7cm]{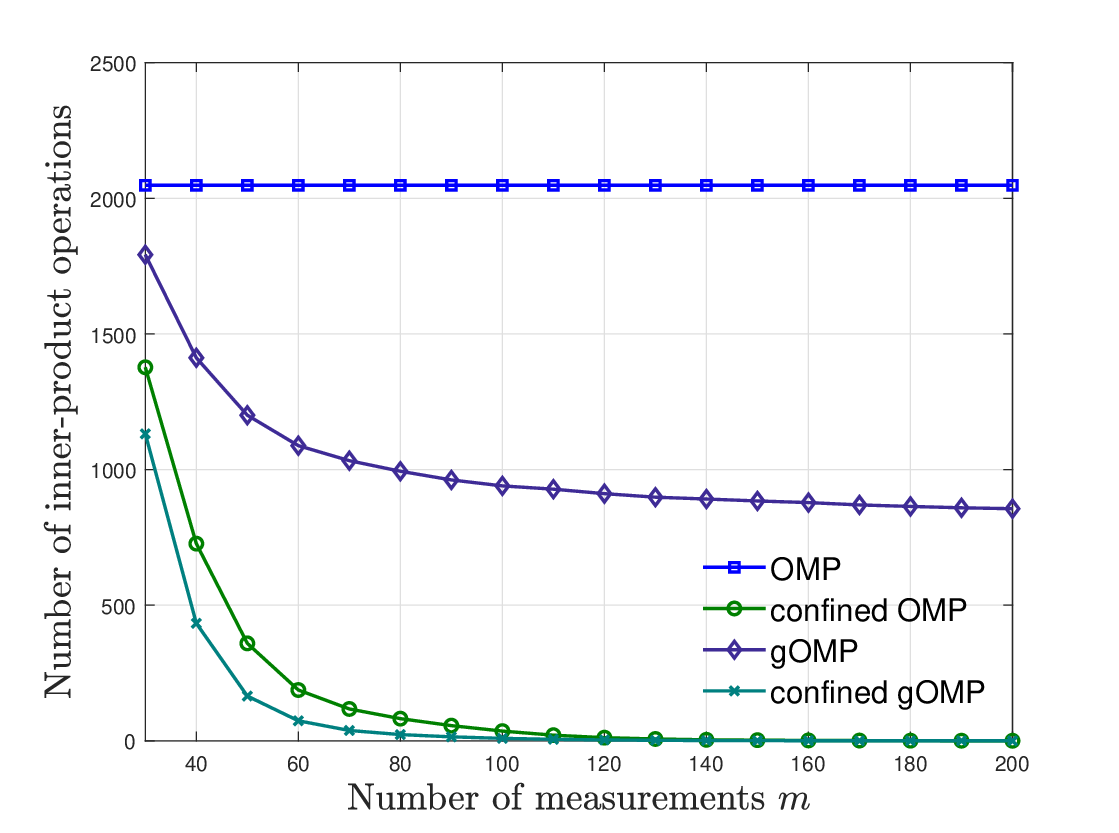}
	}
	\subfigure[]{
		\includegraphics[width=7cm]{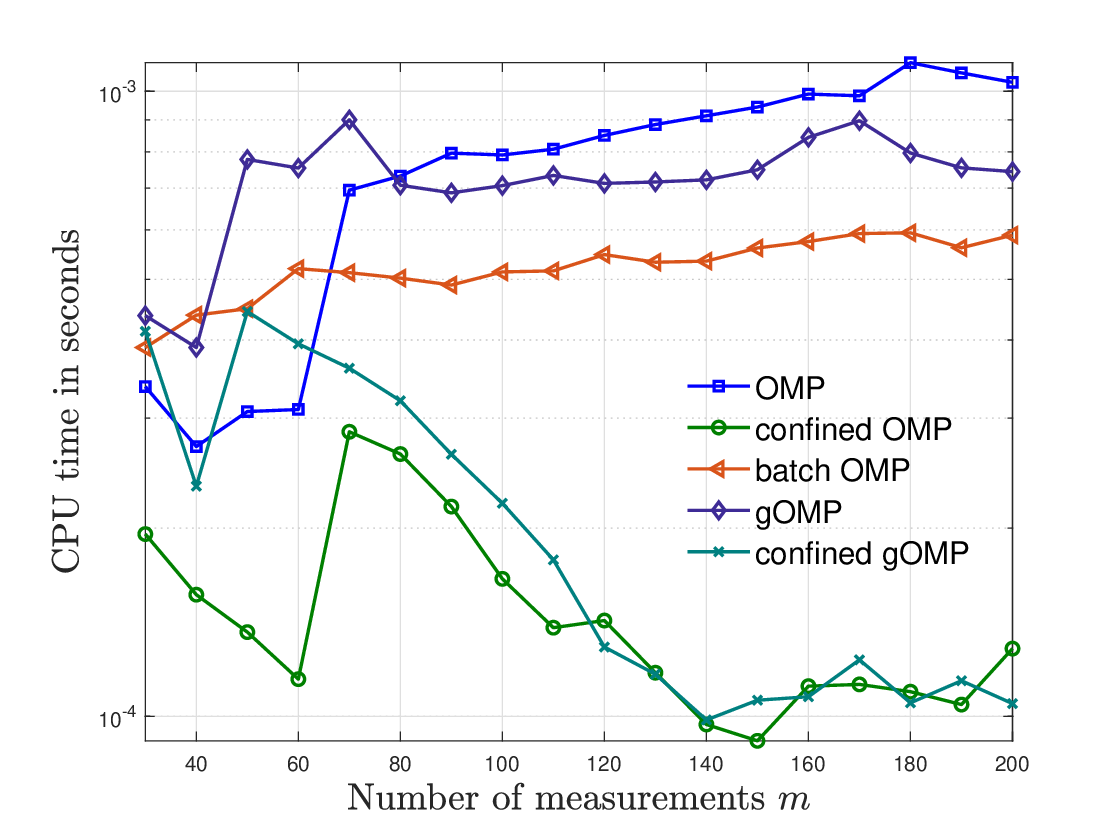}
	}
	\caption{(a)~Average number of inner-product operations in the identification step and (b) average CPU times of different greedy algorithms for Gaussian sparse signals with different measurements $m$.  }\label{fig:Sim_4}
\end{figure*}

\subsection{Recovery Performances for Different Measurements $m$}
We further show the efficiency and recovery performance of the confined OMP algorithm for different measurements $m$ in the noiseless case. In~Fig.~\ref{fig:Sim_3} and~Fig.~\ref{fig:Sim_4}, each column of the measurement matrix $\boldsymbol{A}\in \{0,1\}^{m \times 256}$ is chosen independently and uniformly among the vectors with $10$ ones. The sparsity level $K$ is set to $8$. Other parameters are consistent with those in Fig.~\ref{fig:Sim_1} and Fig.~\ref{fig:Sim_2}.

In~Fig.~\ref{fig:Sim_3}, the recovery performances of confined greedy algorithms are better than those of benchmarks, but performance gains become smaller as $m$ increases. However, we observe in~Fig.~\ref{fig:Sim_4} that the number of inner-product operations for confined greedy algorithms tends to 0 as $m$ increases. These observations indicate that the complexity of the confined greedy algorithms are mainly contributed by the complexity of solving least-squares problems when $m \geq 120$.

\subsection{Sensitivity of Sparsity Level Knowledge}
In many applications, the sparsity level $K$ is unknown in practice. Fortunately, as shown in~(\ref{equ:expy_simple}), the expectation of $\nu^{(K)}$ is related to the sparsity $K$ if the signal is a confined signal. The sparsity $K$ is predictable according to the ``sparsity" of $\boldsymbol{y}$. For the flat sparse signal with amplitude $\theta$, in the absence of noise, the sparsity level $K$ can be accurately detected by $ \vecnorm{\boldsymbol{y}}_1 / (d \cdot |\theta|)$. For other confined signals, a simple way is to calculate $\mathbf{E}_\upsilon [ \nu^{(K)} ]$ for different $K$ first. Then, the estimated sparsity level $\hat{K}$ can be evaluated by
\begin{equation}\label{equ:est_K}
  \hat{K} = \underset{K\geq1} \argmin{ \left| \upsilon - \mathbf{E}_\upsilon [ \nu^{(K)} ] \right| },
\end{equation}
where $\mathbf{E}_\upsilon [ \nu^{(K)} ]$ is given in (\ref{equ:expy_simple}) and $\upsilon = m - |\mathcal{E}|$ is the number of ``nonzero" elements for the received sample $\boldsymbol{y}$.

\begin{figure}[t]
  \centering
  \includegraphics[width=7cm]{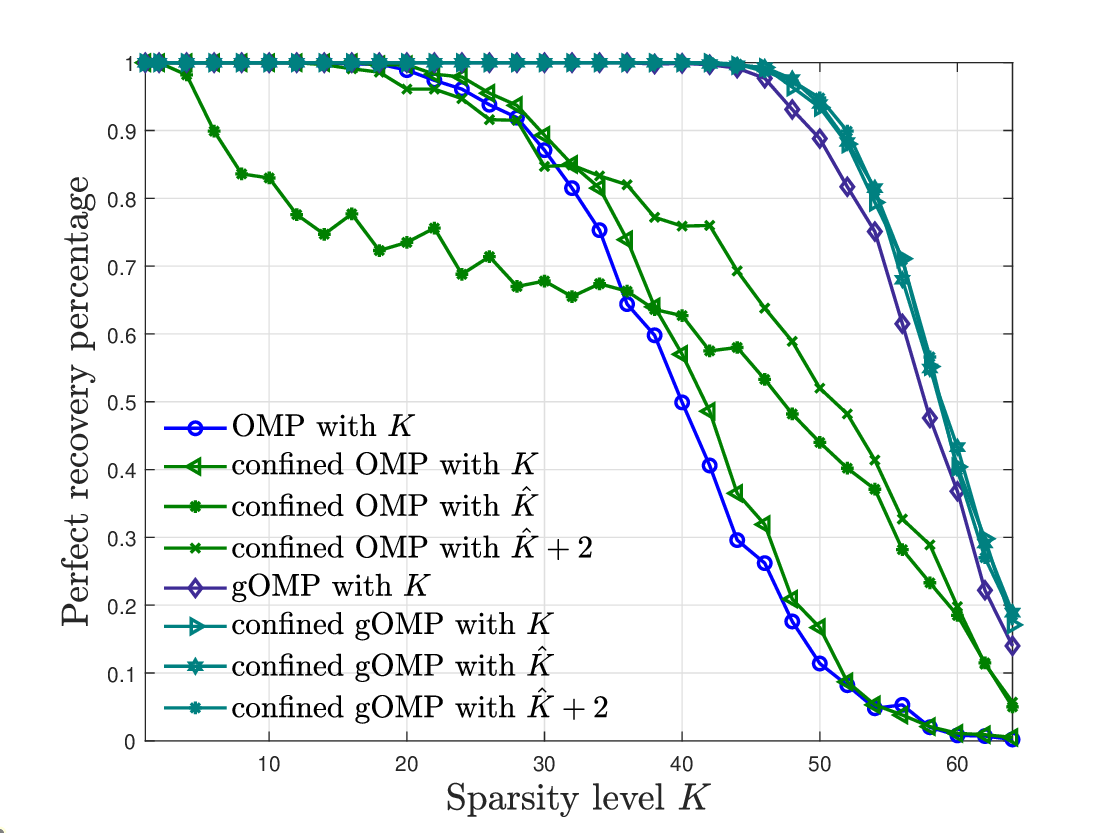}
  \caption{Perfect recovery percentage for Gaussian sparse signals, where the estimated sparsity level $\hat{K}$ is evaluated by (\ref{equ:est_K}).}
  \label{fig:Sim_5}
\end{figure}

%We allows the proposed algorithms to overestimate the estimated sparsity level $\hat{K}$, specifically by setting $\hat{K} = \hat{K} + 2$.

Fig.~\ref{fig:Sim_5} shows the sensitivity of the proposed algorithms to the estimated sparsity level $\hat{K}$. The parameters are consistent with those in~Fig.~\ref{fig:Sim_1}~(a). We see that when $K$ is relatively large, the confined OMP algorithm benefits from an overestimated $\hat{K} > K$ on average, resulting in better recovery performance. Whereas when $K$ is small, the algorithm suffers from an underestimated $\hat{K} < K$ on average. Setting $\hat{K} + 2$ can enhance the recovery performance for smaller $K$ by introducing an additional two iterations. Compared with the confined OMP algorithm, the confined gOMP algorithm shows robustness in recovery performance against the estimated $\hat{K}$, which is attributed to its larger estimated support set.

%In the future, it may be interesting to design a blind recovery algorithm that can achieve the same recovery performance as a non-blind recovery algorithm with a moderate increase in complexity.

\subsection{Robustness in the Noisy Linear System}
In this subsection, we show the robustness of the proposed algorithm in a noisy linear system. Assume that the noise vector $\boldsymbol{v}$ in~(\ref{equ:noise_linear_mod}) is generated randomly, whose elements are independently and randomly drawn from a uniform distribution $U(-\eta,\eta)$. Here, $\eta$ is a positive real number and $\vecnorm{\boldsymbol{v}}_{\infty} \leq \eta$. To maximize redundancy elimination, the $\epsilon$ is set to $\epsilon = \eta$ so that noise-induced perturbations remain within the tolerance range.

\begin{figure}[t]
	\centering
	\subfigure[]{
		\includegraphics[width=7cm]{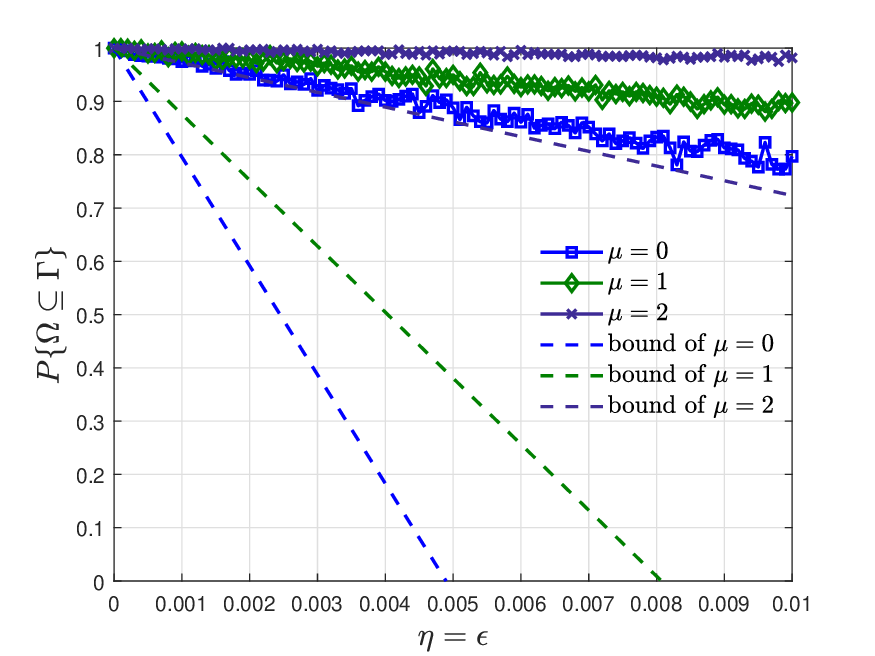}
	}
	\subfigure[]{
		\includegraphics[width=7cm]{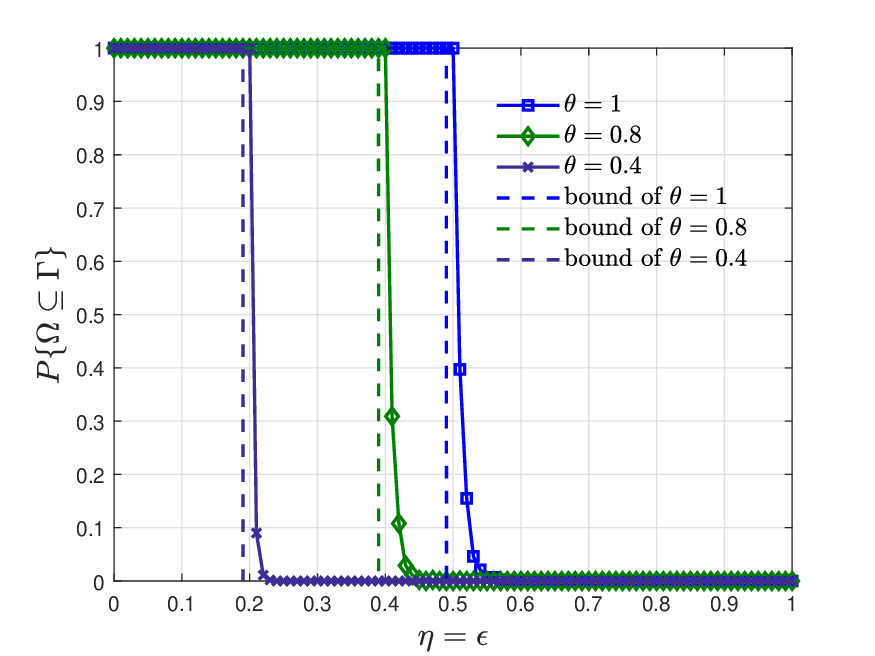}
	}
	\caption{Empirical probability $\mathbb{P} \{ \Omega \subseteq \Gamma \} $ and corresponding lower bounds for (a)~Gaussian sparse signals with means $\mu = \{ 0,1,2 \}$ and (b)~flat sparse signals with amplitudes $\theta = \{0.4,0.8,1 \}$ in a noisy linear system. Note that the bounds are given in~(\ref{equ:conf_noise}). }\label{fig:Sim_6}
\end{figure}

Fig.~\ref{fig:Sim_6} shows the empirical probability $\mathbb{P} \{ \Omega \subseteq \Gamma \} $ for Gaussian and flat sparse signals in a noisy linear system. The parameters are consistent with those in~Fig.~\ref{fig:Sim_1} except for $\mu = \{ 0,1,2 \}$, $\theta = \{0.4,0.8,1 \}$, $\epsilon = \eta$, and $K=10$. It can be seen from Fig.~\ref{fig:Sim_6} that the robustness of the probability $\mathbb{P} \{ \Omega \subseteq \Gamma \} $ to noise is significantly influenced by the probabilistic distribution of the nonzero components of the signal. For the Gaussian sparse signal, increasing the mean $\mu$ can increase the probability $\mathbb{P} \{ \Omega \subseteq \Gamma \} $. The reason can be found in Sec.~\ref{sec:noisy_mod}. Similarly, for the flat sparse signal, increasing the amplitude $\theta$ enhances the robustness of the confined set $\Gamma$ against noise. If $2\eta \leq \theta$, the support of $\boldsymbol{x}$ is a subset of $\Gamma$ with probability 1.

\begin{figure}[t]
  \centering
  \includegraphics[width=7cm]{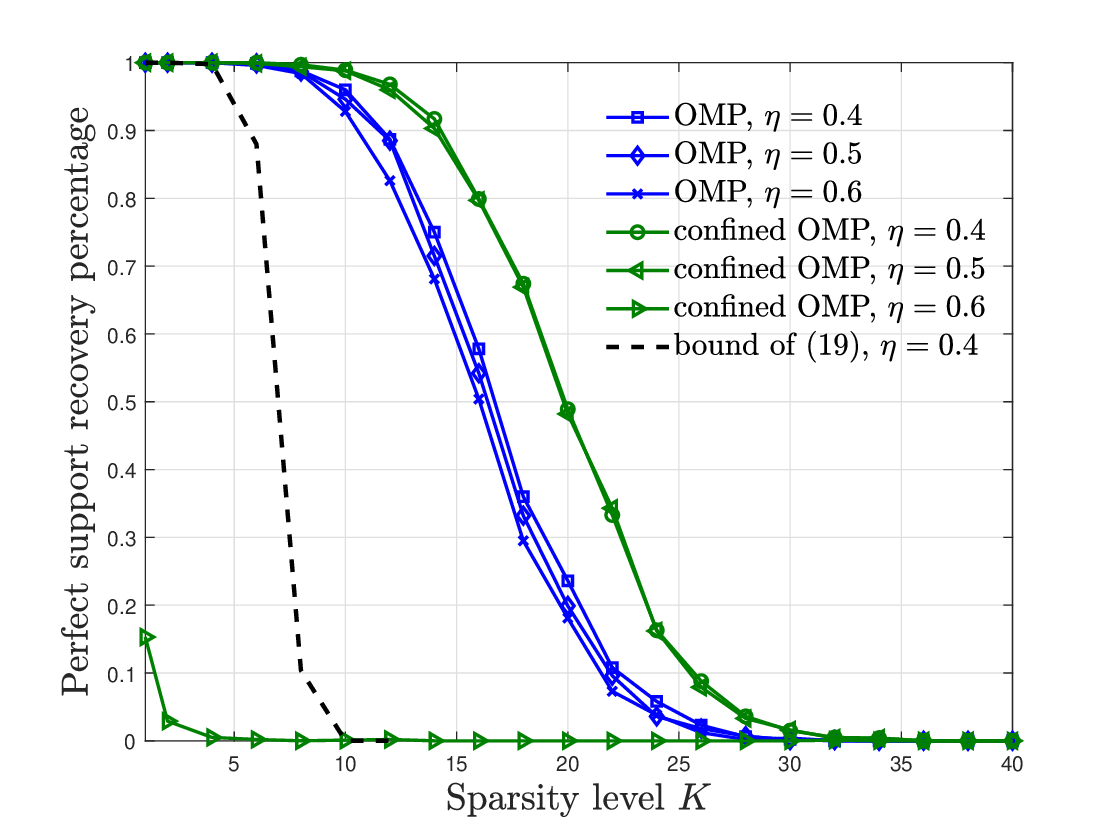}
  \caption{Perfect support recovery percentage for the flat sparse signal with amplitude $\theta = 1$ in a noisy linear system, where the noise vector $\boldsymbol{v}$ satisfies $\vecnorm{\boldsymbol{v}}_{\infty} \leq \eta$. }
  \label{fig:Sim_7}
\end{figure}

Fig.~\ref{fig:Sim_7} shows the perfect support recovery performance for the flat sparse signal with amplitude $\theta = 1$ in a noisy linear system, where $\mathrm{supp}(\hat{\boldsymbol{x}}) = \mathrm{supp}(\boldsymbol{x})$ indicates perfect support recovery. The parameters are consistent with those in~Fig.~\ref{fig:Sim_1}~(b) except for $\epsilon = \eta$. We can see that empirical results shown in Fig.~\ref{fig:Sim_7} are consistent with those in Fig.~\ref{fig:Sim_6}. When $2\eta \leq \theta$, the exact support recovery performance is almost unaffected by noise. Otherwise, the probability $\mathbb{P} \{ \Omega \subseteq \Gamma \} $ drops sharply to 0, and so does the exact support recovery performance.

\subsection{Lower Bounds}
The lower bounds on the probability of $\mathbb{S}_{\mathrm{comp}}$ and $ \mathbb{S}_{\mathrm{comp}}^{(\mathrm{supp})}$ can be obtained by~(\ref{equ:lowerbound_COMP}) and~(\ref{equ:pi_Ksupp}), respectively. In experimental tests, the column degree $d$ of each chosen $\boldsymbol{A}$ satisfies $\ln m < d \leq m/2$ for a sufficiently large $m$. In this case, the spark of $\boldsymbol{A}$ is greater than $K$ with high probability. Thus, the effect of $\mathbb{P} \{ spark(\boldsymbol{A}) > K \}$ is negligible. Furthermore, $\bar{\pi}_K$ shown in~(\ref{equ:lower_boundlooser}) is a looser lower bound on $\mathbb{P} \{ \mathbb{S}_{\mathrm{comp}} \}$. In particular, given $m$ and $d$, the probability $\bar{\pi}_K$ is valid only when $K \leq m/d$. For example, in Fig.~\ref{fig:Sim_1}, the $\bar{\pi}_K$ is valid when $K \leq 128/10$. Similarly, given $K$ and $d$ in Fig.~\ref{fig:Sim_3}, the $\bar{\pi}_K$ is valid when $m \geq K \cdot d = 80$.

The lower bound on $\mathbb{P} \{\mathbb{S}_{\mathrm{comp}} \}$ and the bound $\bar{\pi}_K$ are shown in~Fig.~\ref{fig:Sim_1} and~Fig.~\ref{fig:Sim_3}. It can be seen that there exists a non-negligible gap between theory and practice. The reason is that the bound given in~(\ref{equ:lowerbound_COMP}) is actually the probability $\mathbb{P} \{ |\Gamma| = K \}$. Thus, the gap between the analytical curve and experimental curve is equal to the exact recovery probability conditioned on the event $\{ |\Gamma| > K \}$.

\emph{Remark~4:} In~\cite{Tropp2007Signal} and~\cite{Wen2020Signal}, lower bounds on the exact recovery probability of OMP algorithm over Gaussian matrices also are loose. However, these theoretical results may provide a guide to determine whether the greedy algorithms are appropriate for reconstruction of sparse signals~\cite{Tropp2007Signal}. Otherwise, another sparse recovery algorithm is considered instead.

In Fig.~\ref{fig:Sim_6}, we show the lower bounds of the probability $\mathbb{P} \{ \Omega \subseteq \Gamma \}$ for different confined signals. These bounds, derived from (\ref{equ:conf_noise}), are valid only if they are non-negative. We can see that the trends of these theoretical results are consistent with those of empirical results, which can provide guidance for evaluating the robustness of such confined signals to noise. Fig.~\ref{fig:Sim_7} shows the lower bound on $\mathbb{P} \{ \mathbb{S}_{\mathrm{comp}}^{(\mathrm{supp})} \}$ for the flat sparse signal. Similar to that of (\ref{equ:lower_boundlooser}), this bound, derived from (\ref{equ:pi_Ksupp}), is valid only when $K \leq m/d$. For $\eta = 0.5$ and $\eta = 0.6$, the derived lower bounds are zero and thus not shown in Fig.~\ref{fig:Sim_7}.

\begin{figure}[t]
  \centering
  \includegraphics[width=7cm]{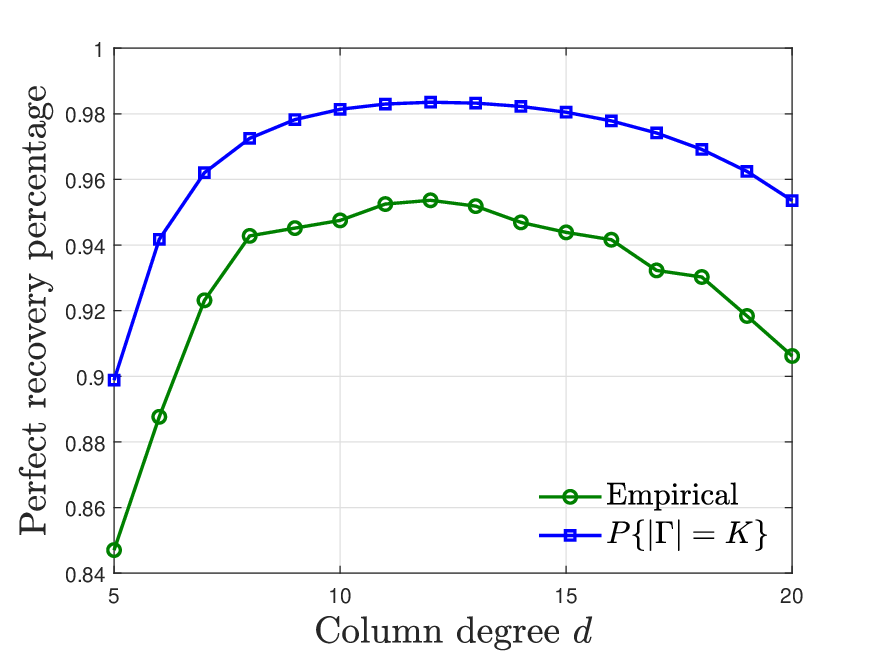}
  \caption{Empirical and theoretical recovery performances of the confined OMP algorithm as a function of column degree $d$.}
  \label{fig:d_opt}
\end{figure}

\subsection{Optimization of Column Degree $d$}
As shown in~(\ref{equ:lowerbound_COMP}), the lower bound on $\mathbb{P} \{\mathbb{S}_{\mathrm{comp}} \}$ is related to the column degree $d$ of $\boldsymbol{A}$. We expect to optimize the recovery performance of confined OMP algorithm by optimizing the lower bound on $\mathbb{P} \{\mathbb{S}_{\mathrm{comp}} \}$. In other words, the optimized target is to increase the probability $\mathbb{P} \{ |\Gamma| = K \}$. In the following experimental simulations, the column degree $d$ of a measurement matrix $\boldsymbol{A} \in \{0,1\}^{100\times 256}$ satisfies $ d > \log m$ such that the probability that the least squares have a unique solution is approximately equal to 1. The experimental and theoretical simulations are based on flat sparse signals with $K=10$ and $K=5$, respectively. It can be seen from~Fig.~\ref{fig:d_opt} that the trend of theoretical curve is almost the same with that of experimental curve. Both of them achieve the best recovery performance when $d = 12$.

\section{Conclusion}
\label{sec:conclusions}
This paper proposed a variant of OMP algorithm, referred to as confined OMP algorithm, to recover a class of sparse signals. We proved that the support of $\boldsymbol{x}$ is contained in the confined set $\Gamma$ if the signal $\boldsymbol{x}$ is defined in~Definition~\ref{def:conf_signal}. We further presented the expectation of $|\Gamma|$ to show that much redundancy of $\boldsymbol{A}$ can be removed, resulting in an improvement of the proposed algorithm in terms of identification efficiency. We also developed lower bounds on the probability $\mathbb{P} \{\mathbb{S}_{\mathrm{comp}} \}$ over sparse random combinatorial matrices. Theoretical results showed that $m = 2\e K \ln (n - K)$ measurements are sufficient to ensure the probability of recovering a $K$-sparse confined signal is at least $1 - \frac{1}{n-K}$. We further investigated the robustness of the proposed algorithm to noise in a noisy linear system. The results indicate that the robustness of the proposed algorithm to noise strongly relies on the distribution of the nonzero elements of the signal. Finally, experimental results demonstrated that confined greedy algorithms outperform their benchmarks in both recovery performance and complexity.

\appendices

\section{The Proof of~Theorem~\ref{thm:conf}}\label{sec:app0}
\begin{proof}
The complement of the event $\left\{ \Omega \subseteq \Gamma \right\}$ is that there exists at least one row $i \in [m]$ such that there exists $j \in \Omega$ with $A_{i,j} = 1$ and $|y_i| \leq \epsilon$. For any row $i \in [m]$, we denote by the event $E_i = \{ \mathrm{there}~\mathrm{exists}~j \in \Omega~\mathrm{such}~\mathrm{that}~A_{i,j} = 1~\mathrm{and}~|y_i| \leq \epsilon \}$. The probability $\mathbb{P} \{ \Omega \nsubseteq \Gamma \}$ is equivalent to the probability $\mathbb{P} \{ \bigcup_{i=1}^m E_i \}$.  Then, by the union bound, we have
\begin{equation}
  \mathbb{P} \{ \Omega \nsubseteq \Gamma \} = \mathbb{P} \left\{ \bigcup_{i=1}^m E_i \right\} \leq \sum_{i = 1}^{m} \mathbb{P} \{ E_i \}.
\end{equation}

Now, we calculate the probability $\mathbb{P} \{ E_i \}$ first. Let $L_i$ denotes the number of $A_{i,j} = 1$  for $j \in \Omega$ in row $i$. By the law of total probability, we have
\begin{equation}
  \mathbb{P} \{ E_i \} = \sum_{\ell = 1}^{K} \mathbb{P} \{ E_i | L_i = \ell \} \mathbb{P} \{ L_i = \ell \}.
\end{equation}
For $j \in \Omega$, the probability $ \mathbb{P} \{ A_{i,j} = 1 \}$ is $d/m$, independent of other columns. Furthermore, for $\ell = 1,2,\cdots, K$, the probability $\mathbb{P} \{ L_i = \ell \}$ follows the binomial distribution with $K$ and $\frac{d}{m}$. Thus, the probability $\mathbb{P} \{ L_i = \ell \} $ is given by
\begin{equation}
\mathbb{P} \{ L_i = \ell \} = \binom{K}{\ell} \left(\frac{d}{m} \right)^{\ell} \left(1 - \frac{d}{m} \right)^{K- \ell}.
\end{equation}

It is known that the $K$ nonzero components of $\boldsymbol{x}$, denoted without loss of generality by $X_1,X_2,\cdots,X_K$, are i.i.d., with the same CDF $F_X(x)$ and PDF $f_X(x)$. Given $L_i = \ell$, the conditional probability $\mathbb{P} \{ E_i | L_i = \ell \}$ is equivalent to the absolute of the sum of $\ell$ nonzero components being less than or equal to $\epsilon$. Without loss of generality, we assume that $Y = X_1 + X_2 + \cdots + X_{\ell}$. Consider the simple case that $\ell = 2$, the CDF of $Y = X_1 + X_2$ is calculated as
\begin{subequations}
 \begin{align}
 F_Y(y) &= \mathbb{P} \{Y \leq y \}=\mathbb{P} \{X_1+X_2 \leq y \}\\
 &= (F_X * F_X)(y) = \int_{-\infty}^{+\infty} f_X(\tau) F_X(y - \tau) \, d \tau.
 \end{align}
\end{subequations}
Similarly, extending $Y$ to the sum of multiple components, its CDF $F_Y(y)$ is given by
\begin{equation}
F_Y(y) = (\underbrace{F_X * F_X * \cdots * F_X}_{\ell \text{ times}})(y) = F_X^{*\ell}(y).
\end{equation}
Once the CDF $F_Y(y)$ is obtained, one can calculate the conditional probability $\mathbb{P} \{ E_i | L_i = \ell \}$:
\begin{equation}
\mathbb{P} \{ E_i | L_i = \ell  \} = \left( F_Y(\epsilon) - F_Y(-\epsilon) \right)
= ( F_X^{*\ell}(\epsilon) - F_X^{*\ell}(-\epsilon) ).
\end{equation}
Thus, the probability $\mathbb{P} \{ E_i \}$ is given by
\begin{equation}\label{equ:Ei_pro}
 \begin{split}
  \mathbb{P} &\{ E_i \} = \\
   &\sum_{\ell = 1}^{K} \binom{K}{\ell} \left( \frac{d}{m} \right)^{\ell} \left( 1 - \frac{d}{m} \right)^{K - \ell} \Big( F_X^{*\ell}(\epsilon) - F_X^{*\ell}(-\epsilon) \Big).
 \end{split}
\end{equation}
The probability $\mathbb{P} \{ \Omega \subseteq \Gamma \}$ is equivalent to the probability $1 - \mathbb{P} \{ \Omega \nsubseteq \Gamma \}$. Thus, we get~(\ref{equ:conf}).
\end{proof}

\section{The Proof of~Lemma~\ref{Lemm:pro_mu}}\label{sec:app1}
\begin{proof}
It is known that $\boldsymbol{y}$ is a linear combination of $K$ columns of $\boldsymbol{A}$. For $K=1$, the value of $\nu^{(1)}$ is always equal to $d$ since the degree of each column of $\boldsymbol{A}$ is $d$. As a result, we have
$$\mathbb{P} \{ \nu^{(1)} = d \} = 1.$$

For $k = 2,3,\cdots,K$, the value of $\nu^{(k)}$ ranges from $d$ to $\min \{kd,m\}$. Here, the value of $\nu^{(k)}$ is equivalent to the number of nonzero elements for a vector resulting from the element-wise OR operation applied to $k$ columns of $\boldsymbol{A}$. It is hard to calculate the probability of $\nu^{(k)}$ directly since it is related to $k$ columns of $\boldsymbol{A}$. Fortunately, it can be modeled as a Markov model. Specifically, by the law of total probability, the probability $\mathbb{P}(\nu^{(k)})$ is
\begin{equation}\label{equ:Markov}
\mathbb{P}(\nu^{(k)}) = \sum_{\nu^{(k-1)}} \mathbb{P}(\nu^{(k)} | \nu^{(k-1)} )\mathbb{P}(\nu^{(k-1)}).
\end{equation}

The transition probability $\mathbb{P}(\nu^{(k)} | \nu^{(k-1)} )$ only depends on $\nu^{(k-1)}$ and $\nu^{(k)}$. This is equivalent to a linear combination of a vector with $\nu^{(k-1)}$ nonzero elements and any column of $\boldsymbol{A}$. If $\nu^{(k-1)} < \max \{\nu^{(k)}-d,d \}$ and $\nu^{(k-1)} > \min \{ \nu^{(k)}, (k-1)d \}$, the transition probability $\mathbb{P}(\nu^{(k)} | \nu^{(k-1)} )$ is obviously 0. Otherwise, the transition probability is
\begin{equation}\label{equ:tran_pro}
\mathbb{P}(\nu^{(k)} | \nu^{(k-1)} ) = \frac{ \binom{\nu^{(k-1)}}{\nu^{(k)}-\nu^{(k-1)}} \binom{m-\nu^{(k-1)}}{d-\nu^{(k)}+\nu^{(k-1)}} }{ \binom{m}{d} },
\end{equation}
where $\nu^{(k)} \in \{ d,d+1,\cdots, \min \{kd,m\} \}$ and $\nu^{(k-1)} \in \{ \max \{\nu^{(k)}-d,d\}, \max \{\nu^{(k)} -d,d \}+1,\cdots,\min \{ \nu^{(k)}, (k-1)d \} \}$. Here,~ (\ref{equ:tran_pro}) represents the probability that any column of $\boldsymbol{A}$ shares exactly $\nu^{(k)}-\nu^{(k-1)}$ nonzero positions with a vector having $\nu^{(k-1)}$ nonzero elements.

The probability $\mathbb{P}(\nu^{(k-1)})$ can be calculated by~(\ref{equ:Markov}) recursively. As a consequence, we get~(\ref{equ:pro_nu}).
\end{proof}

\section{The Proof of Theorem~\ref{thm:exp_Gamma}}\label{sec:app2}
\begin{proof}
As shown in~Theorem~\ref{thm:conf}, the support of $\boldsymbol{x}$ is a subset of the confined set $\Gamma$ with probability 1. Thus, there are at least $K$ out of $n$ columns whose indices are in $\Gamma$. That is, we have $\mathbf{E}[|\Gamma|] \geq K$.

Assume that $|\Gamma| > K$ and there exists a column $\boldsymbol{A}_j$ for $j \in \Gamma \backslash \Omega$. Given $\nu^{(K)} = \upsilon$, the probability of the event $\left\{ j \in \Gamma \backslash \Omega \big| \nu^{(K)} = \upsilon \right\}$ is given by
\begin{equation}\label{equ:gamma_pro}
\mathbb{P} \left\{ j \in \Gamma \backslash \Omega \Big| \nu^{(K)} = \upsilon \right\} = \frac{\binom{\upsilon}{d}}{\binom{m}{d}}.
\end{equation}
Here, (\ref{equ:gamma_pro}) indicates the probability that the column $\boldsymbol{A}_j$ shares exactly $d$ nonzero positions with $\boldsymbol{y}$ having $\upsilon$ nonzero elements. If this is not the case, then $j \notin \Gamma$ since the $i$-th element of $\boldsymbol{A}_j$ must be 0 for $i \in \mathcal{E}$.

Following from~Lemma~\ref{Lemm:pro_mu}, the value of $\upsilon$ ranges from $d$ to $\min \{Kd,m\}$. Thus, by the law of total probability, we have
\begin{equation}\label{equ:gamma_pro_app}
 \begin{split}
\mathbb{P} & \left\{ j \in \Gamma \backslash \Omega \right\} = \\
&\sum_{\upsilon = d}^{\min \{Kd,m\}} \mathbb{P} \left\{ j \in \Gamma \backslash \Omega \Big| \nu^{(K)} = \upsilon \right\} \mathbb{P} \left\{ \nu^{(K)} = \upsilon \right\}\\
&= \sum_{\upsilon = d}^{\min \{Kd,m\}} \frac{\binom{\upsilon}{d}}{\binom{m}{d}} \mathbb{P} \left\{ \nu^{(K)} = \upsilon \right\}\\
&= \mathbf{E}_{\upsilon} \left[ \frac{\binom{\nu^{(K)}}{d}}{\binom{m}{d}} \right],
 \end{split}
\end{equation}
where $\mathbb{P} \left\{ \nu^{(K)} = \upsilon \right\}$ is given in~(\ref{equ:pro_nu}) and $\mathbf{E}_{\upsilon} [ \cdot ]$ is given in~(\ref{equ:exp_y}).

In summary, there are $K$ columns whose indices are in $\Omega \subseteq \Gamma$.  For the remaining $n-K$ column indices, each of them belongs to $\Gamma \backslash \Omega$ with probability $\mathbf{E}_{\upsilon} \left[ \frac{\binom{\nu^{(K)}}{d}}{\binom{m}{d}} \right]$. As a result, we get $\mathbf{E} [|\Gamma|]$ shown in~(\ref{equ:exp_Gamma}).
\end{proof}

\section{The proof of Theorem~\ref{thm:lowerbound_comp}}\label{sec:app5}
\begin{proof}
By the law of total probability, we have
\begin{subequations}
\begin{align}
\mathbb{P} & \{ \mathbb{S}_{\mathrm{comp}} \} = \mathbb{P} \left\{ \mathbb{S}_{\mathrm{comp}} \big| |\Gamma| = K \right\} \mathbb{P} \left\{ |\Gamma| = K \right\} \nonumber\\
&+ \mathbb{P} \left\{ \mathbb{S}_{\mathrm{comp}} \big| |\Gamma| > K \right\} \mathbb{P} \left\{ |\Gamma| > K \right\}\nonumber\\
&= 1 \cdot \mathbb{P} \left\{ |\Gamma| = K \right\} + \mathbb{P} \left\{ \mathbb{S}_{\mathrm{comp}} \big| |\Gamma| > K \right\} \mathbb{P} \left\{ |\Gamma| > K \right\}\label{equ:app5_1a}\\
&\geq \mathbb{P} \left\{ |\Gamma| = K \right\}.\label{equ:app5_1b}
\end{align}
\end{subequations}
As shown in~Algorithm~\ref{alg:COMP}, the identification is already done if $|\Gamma| = K$. Thus, the probability $\mathbb{P} \left\{ \mathbb{S}_{\mathrm{comp}} \big| |\Gamma| = K \right\}$ is reduced to the probability that the least squares has a unique solution. Since $\boldsymbol{A}$ satisfies $spark(\boldsymbol{A}) > K$, any $K$ columns of $\boldsymbol{A}$ are linearly independent. In other words, $\boldsymbol{A}_{\Lambda^{(K)}}^{T} \boldsymbol{A}_{\Lambda^{(K)}}$ is of full rank. Thus, we have $\mathbb{P} \left\{ \mathbb{S}_{\mathrm{comp}} \big| |\Gamma| = K \right\} = 1$. In~(\ref{equ:app5_1a}), calculating $\mathbb{P} \left\{ \mathbb{S}_{\mathrm{comp}} \big| |\Gamma| > K \right\}$ is complicated. Hence, we only consider the contribution of the first term of~(\ref{equ:app5_1a}) for simplicity, resulting in~(\ref{equ:app5_1b}).

Given any $j \in \Omega^c$, the probability $\mathbb{P} \left\{ |\Gamma| = K \right\}$ is equivalent to the probability $\mathbb{P} \{ j \notin \Gamma \}^{|\Omega^c|}$. It is known that the value of $\upsilon$ ranges from $d$ to $\min \{Kd,m\}$. Given $\nu^{(K)} = \upsilon$, the probability $\mathbb{P} \left\{ j \in \Gamma \backslash \Omega \big| \nu^{(K)} = \upsilon \right\}$ is given in~(\ref{equ:gamma_pro}). Then, we have
\begin{subequations}
\begin{align}
\mathbb{P} \{ |\Gamma| = K \} &= \mathbb{P} \{ j \notin \Gamma \}^{|\Omega^c|}\nonumber \\
&= \left( 1 - \mathbb{P} \left\{ j \in \Gamma \backslash \Omega \right\} \right)^{n-K}\\
&= \sum_{\upsilon =d}^{\min \{Kd,m \}} \left( 1 - \frac{\binom{\upsilon}{d}}{\binom{m}{d}} \right)^{n-K} \mathbb{P} \{ \nu^{(K)} = \upsilon \},\\
&= \mathbf{E}_{\upsilon} \left[  \left(1 - \frac{\binom{\nu^{(K)}}{d}}{\binom{m}{d}} \right)^{n-K}  \right],
\end{align}
\end{subequations}
where $\mathbb{P} \{ \nu^{(K)} = \upsilon \}$ is given in~(\ref{equ:pro_nu}) and $\mathbf{E}_{\upsilon} [ \cdot ]$ is given in~(\ref{equ:exp_y}). Thus, we get (\ref{equ:lowerbound_COMP}).
\end{proof}
%\section*{Acknowledgments}
%We would like to acknowledge the assistance of volunteers in putting
%together this example manuscript and supplement.

\section{The proof of Corollary~\ref{cor:lower_boundlooser}}\label{sec:app6}
\begin{proof}
Since $K \leq m/d$, the value of $\nu^{(K)}$ ranges from $d$ to $Kd$. With the Jensen’s inequality, the probability $\mathbb{P} \{ \mathbb{S}_{\mathrm{comp}} \}$ is lower bounded by
\begin{subequations}
\begin{align}
\mathbb{P} \{ \mathbb{S}_{\mathrm{comp}} \} &\geq  \left( 1 - \mathbf{E}_{\upsilon} \left[ \frac{\binom{\nu^{(K)}}{d}}{\binom{m}{d}} \right]  \right)^{n-K} \\
&\geq \left( 1 - \frac{\binom{Kd}{d}}{\binom{m}{d}} \right)^{n-K} \\
&= \left( 1 - \prod_{z=0}^{d-1} \frac{Kd - z}{m -z} \right)^{n-K} \\
&\geq \left( 1 - \left( \frac{Kd}{m} \right)^d \right)^{n-K} =\bar{\pi}_K.
\end{align}
\end{subequations}

With all parameter in $\bar{\pi}_K$ fixed except for $d$, we derive the maximum value of $\bar{\pi}_K(d)$ with respect to $d$. Let $h(d) = \left( \frac{Kd}{m} \right)^d$. We first need to perform differentiation on $h(d)$. Before that, we respectively take the logarithm of both sides, i.e.,
\begin{equation}
\ln h(d) = d \ln\left( \frac{Kd}{m} \right).
\end{equation}
Then, we have
\begin{subequations}
\begin{align}
\left( \ln h(d) \right)'  &= \left( d \ln\left( \frac{Kd}{m} \right) \right)',\\
\frac{h'(d)}{h(d)}  &=  \ln\left( \frac{Kd}{m} \right) + d \cdot \frac{m}{Kd} \cdot \frac{K}{m},\\
h'(d) &= h(d) \cdot \left( \ln\left( \frac{Kd}{m} \right) + 1 \right).
\end{align}
\end{subequations}
Now, we perform differentiation on $\bar{\pi}_K(d)$. Similarly, we have
\begin{subequations}
\begin{align}
\ln \bar{\pi}_K(d) &= (n-K)\ln \left( 1 - h(d) \right),\\
\left( \ln \bar{\pi}_K(d) \right)' &= \left( (n-K)\ln \left( 1 - h(d) \right) \right)',\\
\bar{\pi}_K'(d) &= -\bar{\pi}_K(d) \cdot (n-K) \frac{h'(d)}{1-h(d)},\\
\bar{\pi}_K'(d) &= -\bar{\pi}_K(d) \cdot (n-K) \frac{h(d)}{1-h(d)} \left( \ln\left( \frac{Kd}{m} \right) + 1 \right).
\end{align}
\end{subequations}
Let $\bar{\pi}_K'(d) = 0$. Since $\bar{\pi}_K(d) > 0$, $h(d) > 0$, and $n-K > 0$, $\bar{\pi}_K'(d) = 0$ is equivalent to $ \ln\left( \frac{Kd}{m} \right) + 1 = 0$. A simple derivation can yield that $\bar{\pi}_K(d)$ reaches its maximum value when $d= \frac{m}{K \cdot \e}$. By substituting $d = \frac{m}{K\cdot \e}$ into $\bar{\pi}_K(d)$, we get $\bar{\pi}_K^{*} = \left( 1 -  \e^{-\frac{m}{K\cdot \e}} \right)^{n-K}$.
\end{proof}

\section{The proof of Theorem~\ref{thm:nec_m}}\label{sec:app7}
\begin{proof}
As indicated in~(\ref{equ:lower_boundlooser}), the recovery probability $\mathbb{P} \{ \mathbb{S}_{\mathrm{comp}} \}$ has a lower bound of $\mathbb{P} \{ \mathbb{S}_{\mathrm{comp}} \} \geq \bar{\pi}_K$. By substituting $d = \frac{m}{K \beta}$ into $\bar{\pi}_K(d)$, we obtain the following function:
\begin{equation}
  \bar{\pi}_K(\beta) = \left( 1 - \beta^{-\frac{m}{K \beta}} \right)^{n-K},
\end{equation}
where the term $\beta^{-\frac{m}{K \beta}}$ is strictly less than 1. With the Bernoulli’s inequality, we have
\begin{equation}
  \bar{\pi}_K(\beta) \geq 1 - \left( n-K \right) \beta^{-\frac{m}{K \beta}}.
\end{equation}
For any constant $c > 0$, let $m = c K \log_{\beta}(n-K)$. Then, we have
\begin{equation}
 \bar{\pi}_K(\beta) \geq 1 - \left( n-K \right) \beta^{-\frac{c K \log_{\beta}(n-K)}{K \beta}} = 1 - (n-K)^{1-c/\beta}.
\end{equation}
To ensure $\bar{\pi}_K(\beta) > 0$, the inequality $c > \beta$ must be satisfied. In other words, the measurements $m$ should satisfy
\begin{equation}
  m = c K \log_{\beta}(n-K) > \frac{\beta}{\ln \beta} K \ln (n - K).
\end{equation}
\end{proof}

\section{The proof of Corollary~\ref{cor:asymptotic}}\label{sec:app10}
\begin{proof}
With the Bernoulli’s inequality, we have
\begin{equation}
  \bar{\pi}_K \geq 1 - \left( n-K \right) \left( \frac{Kd}{m} \right)^d \geq 1 - n\left( \frac{Kd}{m} \right)^d.
\end{equation}
By substituting $d = \gamma \ln m$ and $n = m^\tau$ into the above inequality, we obtain
\begin{subequations}
\begin{align}
\bar{\pi}_K &\geq 1 - m^\tau \left( \frac{K\gamma \ln m}{m} \right)^{\gamma \ln m}\\
& = 1 - \e^{\tau \ln m} \cdot \e^{\gamma \ln m \ln \frac{K\gamma \ln m}{m}}\\
& = 1 - \e^{\tau \ln m + \gamma \ln m \ln \frac{K\gamma \ln m}{m}}.
\end{align}
\end{subequations}
To ensure $\bar{\pi}_K$ converges to 1 as $m$ goes to infinity, we have
\begin{equation}
  \underset{m \to \infty}\lim \ln m \left( \tau + \gamma \ln \frac{K\gamma \ln m}{m} \right) \to -\infty.
\end{equation}
That is,
\begin{equation}
   \tau + \gamma \ln \frac{K\gamma \ln m}{m} < 0.
\end{equation}
Finally, after a simple derivation, we can obtain
\begin{equation}
   K < \frac{1}{\gamma} \cdot \e^{-\tau/\gamma}  \cdot \frac{m}{\ln m}.
\end{equation}
The proof of finding the maximum value is omitted.

\end{proof}

\section{The proof of Theorem~\ref{thm:noise_confset}}\label{sec:app8}
\begin{proof}
Without loss of generality, we denote by $X_1,X_2,\cdots,X_K$ the $K$ nonzero components of $\boldsymbol{x}$, respectively. For any row $i \in [m]$, we denote by the event $E_i = \{ \mathrm{there}~\mathrm{exists}~j \in \Omega~\mathrm{such}~\mathrm{that}~A_{i,j} = 1~\mathrm{and}~|y_i| \leq \epsilon \}$. Let $L_i$ denotes the number of $A_{i,j} = 1$  for $j \in \Omega$ in row $i$. Then, we have
\begin{equation}
  \mathbb{P} \{ E_i \} = \sum_{\ell = 1}^{K} \mathbb{P} \{ E_i | L_i = \ell \} \mathbb{P} \{ L_i = \ell \}.
\end{equation}
The proof process is similar to that of Theorem~\ref{thm:conf}. The main difference is that the conditional probability $\mathbb{P} \{ E_i | L_i = \ell \}$ is equivalent to the absolute of the sum of $\ell$ nonzero components plus the noise being less than or equal to $\epsilon$. That is, we have
\begin{equation}
  \mathbb{P} \{ E_i | L_i = \ell \} = \mathbb{P} \{ |X_1+X_2+\cdots+X_\ell + v_i| \leq \epsilon \}.
\end{equation}
With $\vecnorm{\boldsymbol{v}}_{\infty} \leq \eta$, the probability $\mathbb{P} \{ E_i | L_i = \ell \}$ is less than or equal to the probability $\mathbb{P} \{ |X_1+X_2+\cdots+X_\ell| \leq \epsilon + \eta \}$. With~(\ref{equ:Ei_pro}) and~(\ref{equ:expy_simple}), we have (\ref{equ:app8_1}) below.

\begin{figure*}[b]
    \centering
    \vspace*{8pt}
    \hrulefill
    \vspace*{8pt}
\begin{subequations}
\begin{align}
\mathbb{P} \{ E_i \} &\leq \sum_{\ell = 1}^{K} \binom{K}{\ell} \left( \frac{d}{m} \right)^{\ell} \left( 1 - \frac{d}{m} \right)^{K - \ell} \Big( F_X^{*\ell}(\epsilon + \eta) - F_X^{*\ell}(-\epsilon - \eta) \Big) \label{equ:bound_noisyy} \\
&\leq \sum_{\ell = 1}^{K} \binom{K}{\ell} \left( \frac{d}{m} \right)^{\ell} \left( 1 - \frac{d}{m} \right)^{K - \ell} \cdot \underset{\ell \in [K]} \max \left( F_X^{*\ell}(\epsilon + \eta) - F_X^{*\ell}(-\epsilon - \eta) \right) \\
&= \left( 1 - \left( 1 - \frac{d}{m} \right)^K \right) \cdot \underset{\ell \in [K]}\max \left( F_X^{*\ell}(\epsilon + \eta) - F_X^{*\ell}(-\epsilon - \eta) \right) \label{equ:app8_1}
\end{align}
\end{subequations}
    \vspace*{8pt}
    \hrulefill
\end{figure*}

With $\mathbb{P} \{ \Omega \subseteq \Gamma \} \geq 1 - \sum_{i = 1}^{m} \mathbb{P} \{ E_i \}$, we get~(\ref{equ:conf_noise}).
%\begin{subequations}
%\begin{align}
%\mathbb{P} \{ E_i \} &= \sum_{\ell = 1}^{K} \binom{K}{\ell} \left( \frac{d}{m} \right)^{\ell} \left( 1 - \frac{d}{m} \right)^{K - \ell} \Big( F_X^{*\ell}(\epsilon + \eta) - F_X^{*\ell}(-\epsilon - \eta) \Big)\\
%&\leq \sum_{\ell = 1}^{K} \binom{K}{\ell} \left( \frac{d}{m} \right)^{\ell} \left( 1 - \frac{d}{m} \right)^{K - \ell} \underset{\ell \geq 1}\max \left( F_X^{*\ell}(\epsilon + \eta) - F_X^{*\ell}(-\epsilon - \eta) \right)\\
%&= \left( 1 - \left( 1 - \frac{d}{m} \right)^K \right) \cdot \underset{\ell \geq 1}\max \left( F_X^{*\ell}(\epsilon + \eta) - F_X^{*\ell}(-\epsilon - \eta) \right)
%\end{align}
%\end{subequations}
\end{proof}

\ifCLASSOPTIONcaptionsoff

\newpage
\fi

\bibliographystyle{IEEEtran}
\bibliography{IEEEabrv,references}
	
\end{document}